\documentclass[a4paper]{article}

\usepackage{amsmath}
\usepackage{amsthm}
\usepackage{amssymb}
\usepackage{enumerate}
\usepackage{booktabs}
\usepackage{rotating}
\usepackage{verbatim}
\usepackage[numbers]{natbib}
\usepackage{subcaption}
\usepackage{graphicx}
\usepackage{float}
\usepackage{hyphenat}
\usepackage{hyperref}
\usepackage[titletoc,title]{appendix}
\usepackage{units}
\usepackage[affil-it]{authblk}
\usepackage{epstopdf}

\usepackage{tikz,tikz-3dplot}
\usepackage{pgfplots}
\usepackage{pgfplotstable}

\allowdisplaybreaks

\title{Degree Bounded Bottleneck Spanning Trees in Three Dimensions}

\author{Patrick J. Andersen%
  \thanks{Email: \texttt{pat.j.andersen@gmail.com}; Phone: \texttt{+61-413579238}; Corresponding author}
}

\author{Charl J. Ras}

\affil{School of Mathematics and Statistics, The University of Melbourne, Australia}

\newtheorem{lemma}{Lemma}
\newtheorem{theorem}{Theorem}
\newtheorem{proposition}{Proposition}
\newtheorem{corollary}{Corollary}

\newtheorem{conjecture}{Conjecture}

\providecommand{\keywords}[1]{\textit{Keywords: } #1}

\floatstyle{ruled}
\newfloat{algorithm}{htb}{alg}
\floatname{algorithm}{Algorithm}

\makeatletter
\def\@maketitle{%
  \newpage
  \null
  \vskip 2em%
  \begin{center}%
  \let \footnote \thanks
    {\Large\bfseries \@title \par}%
    \vskip 1.5em%
    {\normalsize
      \lineskip .5em%
      \begin{tabular}[t]{c}%
        \@author
      \end{tabular}\par}%
    \vskip 1em%
    {\normalsize \@date}%
  \end{center}%
  \par
  \vskip 1.5em}
\makeatother

\begin{document}
\maketitle
\begin{abstract}
The geometric \textit{$\delta$-minimum spanning tree problem} ($\delta$-MST) is the problem of finding a minimum spanning tree for a set of points in a normed vector space, such that no vertex in the tree has a degree which exceeds $\delta$, and the sum of the lengths of the edges in the tree is minimum. The similarly defined geometric \textit{$\delta$-minimum bottleneck spanning tree problem} ($\delta$-MBST), is the problem of finding a degree bounded spanning tree such that the length of the longest edge is minimum. For point sets that lie in the Euclidean plane, both of these problems have been shown to be NP-hard for certain specific values of $\delta$. In this paper, we investigate the $\delta$-MBST problem in $3$-dimensional Euclidean space and $3$-dimensional rectilinear space. We show that the problems are NP-hard for certain values of $\delta$, and we provide inapproximability results for these cases. We also describe new approximation algorithms for solving these $3$-dimensional variants, and then analyse their worst-case performance.\\\\
\end{abstract}
\section{Introduction}
Spanning tree problems have been studied throughout the history of combinatorial optimisation. They arise in a multitude of practical settings including computer and telecommunication networks, transportation, plumbing, and electrical circuit design \cite{graham1985history}. In graph theoretic terms, a typical spanning tree problem involves finding a spanning subgraph $G'= (V,E')$ of a given graph $G = (V,E)$, where $G'$ is connected and $|E'| = |V| - 1$, such that $G'$ is optimal with respect some objective function on the set of edges. There may also be additional constraints imposed on $G'$, such as constraints on maximum degree or maximum diameter of the tree. In this paper we will focus on geometric versions of the spanning tree problem, where the given graph will be a complete graph whose vertices lie in a normed vector space. Each edge in the graph will have an associated \textit{length} given by the distance between the endpoints of the edge.\\\\  
The most well-known spanning tree problem is the \textit{minimum spanning tree} (MST) problem, where the aim of this problem is to find a spanning tree such that the sum of the lengths of the edges of the tree is minimum. A similar problem is the \textit{bottleneck spanning tree problem}, in which we require that the length of the longest edge in the tree is minimum. Both of these spanning tree problems can be solved efficiently using polynomial time algorithms (see Kruskal \cite{kruskal1956shortest}, Prim \cite{prim1957shortest} for the MST problem, and Camerini \cite{mbst} for the MBST problem). In fact every algorithm which solves the MST problem can be used to solve the MBST problem since every MST is a MBST.\\\\
In this paper, we focus on variants of the MST and MBST where we include constraints on the maximum degree of the spanning trees. In these variants we have an associated degree bound $\delta$ such that no vertex in the optimal spanning tree can have a degree which exceeds $\delta$. We refer to these degree-constrained versions of the MST and MBST problems as the $\delta$-MST problem and $\delta$-MBST problem respectively. The most general versions of the $\delta$-MST problem and $\delta$-MBST problems are NP-hard for all $\delta \geq 2$ \cite{GareyJohnson}. When restricted to the Euclidean plane, both the $\delta$-MST and $\delta$-MBST problems can be solved in polynomial time when $\delta \geq 5$ \cite{degree5EMST}. It has been shown that the $\delta$-MST problem in the Euclidean plane is NP-hard when $\delta = 2,3$ (see \cite{papavazi}), and $4$ (see \cite{francke2009euclidean}), and the $\delta$-MBST problem in the Euclidean plane is NP-hard when $\delta = 2$ and $3$ \cite{andersen2016minimum}. Approximation algorithms for the Euclidean $\delta$-MST problem have been developed by Khuller et al. \cite{khuller1996low} and Chan \cite{chan2003euclidean} and their approximation ratios for the $\delta$-MBST problem were explored in \cite{andersen2016minimum}. The performances of various heuristic and approximation algorithm schemes for the Euclidean $\delta$-MST and $\delta$-MBST problems were investigated in \cite{andersen2018algorithms}.\\\\
\textbf{Our results:} We analyse the $\delta$-MBST problem in $3$-dimensional Euclidean space, which we denote as the $\delta$-E3MBST problem, and in $3$-dimensional rectilinear space, which we denote as the $\delta$-R3MBST problem. We show that the $\delta$-E3MBST and $\delta$-R3MBST problems are NP-hard for $\delta = 4,5$ and we provide inapproximability results for these cases. We describe new approximation algorithms that are simple extensions of the constant factor approximation algorithm of Khuller et al. to 3-dimensional space. We then analyse the performance of these algorithms by attempting to find instances, through both analytic and experimental means, for which the algorithms perform at their worst-case accuracy. We also describe more sophisticated extensions of algorithm of Khuller et al., for which we provide similar analyses of performance. 


\section{Hadwiger Number}\label{sec:Hadwiger}
Let $C$ be a convex body in $\mathbb{R}^d$. The \textit{Hadwiger number} \cite{toth1975hadwiger} (also sometimes referred to in the literature as the \textit{kissing number}) $H(C)$ is the maximum number of mutually non-overlapping translates of $C$ that can be arranged to touch $C$. The \textit{strict Hadwiger number} $H^*(C)$ is defined similarly, with the further restriction that all the translates of $C$ are pairwise disjoint. For example, if $B^2$ is the unit disc in the Euclidean plane, it is known that $H(B^2) = 6$ and $H^*(B^2) = 5$.\\~\\
Let $\delta(v)$ denote the degree of a vertex $v$, and for a tree $T =(V,E)$, let $\delta(T) = \max\{ \delta(v): v \in V\}$. For a finite set of points $S$ in a metric space, let $\delta^+(S) = \max\{\delta(T): T \text{ is a}$ $\text{minimum spanning tree of } S \}$, and let $\delta(S)^- =  \min\{\delta(T): T \text{ is a minimum}$ $\text{spanning tree of } S \}$. For a metric space $M$, let $\delta^+(M) = \max\{\delta^+(S):S \subset M, S \text{ is finite}\}$, i.e., $\delta^+(M)$ is the maximum possible degree that any vertex in an MST of a set of points in $M$ can achieve. Similarly, let $\delta^-(M) = \max\{\delta^-(S):S \subset M, S \text{ is finite}\}$, i.e, $\delta^-(M)$ is the smallest constant $c$, such that any finite point set in $M$ has a spanning tree with maximum degree at most $c$. For example, for the Euclidean plane, which we denote by $E^2$, it is well known that $\delta^+(E^2) = 6$ and that $\delta^-(E^2) = 5$ \cite{degree5EMST}, hence it can be seen that $\delta^+(E^2) = H(B^2)$ and $\delta^-(E^2) = H^*(B^2)$. For any $d$-dimensional Minkowski space $X^d$ with unit ball $B$, Cieslik \cite{cieslik19911} has shown that $\delta^+(X^d) = H(B)$ and Martini and Swanepoel \cite{martini2006low} have shown that $\delta^-(X^d) = H^*(B)$.\\~\\
Since we are focusing on three dimensional Euclidean space, which we denote by $E^3$, and three dimensional rectilinear space, which we denote by $R^3$, we need only consider the Hadwiger numbers of the 3-dimensional Euclidean unit sphere $B^3$ and the 3-dimensional rectilinear unit octahedron $O^3$. It is well known that $H(B^3) = H^*(B^3) = 12$, thus we can conclude that $\delta^+(E^3) = \delta^-(E^3) = 12$. A corollary of this is that any MST is a $\delta$-MST when $\delta \geq 12$, and that there exists a point set $S$ in $E^3$ for which all minimum spanning trees contain a vertex of degree 12. For example, let $S$ contain the origin and all vertices of a regular icosahedron centred at the origin, then $S$ will have a unique spanning tree in which the vertex at the origin has degree 12. \\~\\
Let $\delta$-E3MST denote the $\delta$-MST problem in $E^3$. Since we can find an MST in polynomial time, we can conclude the following.

\begin{theorem}
The $\delta$-E3MST problem can be solved in polynomial time when $\delta \geq 12$. 
\end{theorem}

For $O^3$ and $R^3$, it has been shown that $H(O^3) = 18$ (see \cite{robins1995low}, \cite{talata1999translative},\cite{larman1999kissing}) and that $13 \leq \delta^-(R^3) \leq 14$ \cite{robins1995low}.

\section{Complexity of the $\delta$-E3MST Problem}
In this section we present computational complexity results for the $\delta$-E3MBST and $\delta$-R3MBST problems for some specific cases of $\delta$. Our results extend the complexity results for the $\delta$-E2MBST and $\delta$-R2MBST problems provided in \citep{andersen2016minimum} into three dimensions. The technique used in \citep{andersen2016minimum} was to transform results for the the $\delta$-MST problem in $E^2$, which we refer to as the $\delta$-E2MST problem, into complexity results for the bottleneck version of the problem. In keeping with this method, we will first start with the complexity results for the $\delta$-E2MST problem and extend them to three dimensions before transforming the results into complexity proofs for the bottleneck version.\\~\\
The $2$-E2MST problem is equivalent to the Euclidean Travelling Salesman Path problem, which is known to be NP-hard \cite{papadimitriou1977euclidean},\cite{GareyJohnson}. Papadimitriou and Vazirani \cite{papavazi} have proven that $\delta$-E2MST problem is NP-hard for $\delta = 3$, and Francke and Hoffman \cite{francke2009euclidean} have proven that the problem is NP-hard for $\delta = 4$. Thus the $\delta$-E2MST problem is NP-hard for $\delta \in \{2,3,4\}$ and hence the $\delta$-E3MST problem is NP-hard for these values of $\delta$ since 2-dimensional Euclidean space is a subset of 3-dimensional Euclidean space. To extend the results to the next highest dimension, we will modify the grid graph approach of Papadimitriou and Vazirani.\\~\\ 
We define a \textit{grid graph} to be any finite connected vertex-induced subgraph of the infinite 2-dimensional grid embedded in the plane so that the vertices are at integer coordinates and vertices are adjacent if and only if they have unit distance from each other. We have the following result used by Papadimitriou and Vazirani \cite{papavazi}. 

\begin{theorem} \label{thm:hamiltonian_path_grid_graph}
 The Hamiltonian path problem for 2-dimensional grid graphs with maximum degree 3 is NP-complete.   
\end{theorem}



By extending the approach used in \cite{papavazi}, we can obtain the following result.

\begin{theorem} \label{thm:deg5NPMST}
 The 5-E3MST problem is NP-complete.   
\end{theorem}

\begin{proof}
Suppose we are given a connected grid graph $G= (V,E)$ with maximum degree 3, where we will assume that $|V| \geq 2$. Since $G$ is a grid graph, all its edges have unit length under the Euclidean metric. We perform a 2-colouring of $G$ so that the vertex set $V$ is partitioned into a set of \textit{black} vertices $B$ and \textit{white} vertices $W$ such that $E \subseteq B \times W$ (this is possible since grid graphs are bipartite). For each node $v \in V$, we add a ``pseudo node" $u$ that is a small distance away from $v$ at direction leading towards missing neighbours of $v$ (there will always be such a direction since the degree of $v$ is at most $3$, whereas the maximum degree in a grid graph is $4$). If $v$ is a black vertex, then we let $u$ be a distance of $\epsilon$ from $v$ and refer to these nodes as \textit{black lateral pseudo nodes}; otherwise we let $u$ be a distance of $\delta$ from $v$ and refer to the nodes as \textit{white lateral pseudo nodes}, where $0<\delta<\epsilon<\frac{2-\sqrt{2}}{2} < 1$. Thus, any two black lateral pseudo nodes have distance of least $\sqrt{2}(1-\epsilon) >1$ from each other, any two white lateral pseudo nodes have distance of least $\sqrt{2}(1-\delta)>1$ from each other, and each black lateral pseudo node has a distance of at least $\sqrt{(\epsilon-\delta)^2+1}>1$ from every white lateral pseudo node.\\~\\
In addition to these pseudo nodes, we also add two pseudo nodes $w_1$ and $w_2$ to each vertex $v$ which we will add above and below $v$, assuming the grid graph $G$ lies on the $x$-$y$ plane in 3-dimensional space. If $v$ is a black vertex, then we place $w_1$ directly above $v$ at a distance of $\epsilon'$ away from $v$ and refer to it as a black upper pseudo node, and we place $w_2$ at a distance of $\epsilon'$ directly below $v$ and refer to it as a black lower pseudo node, where $\sqrt{1 - \epsilon^2} < \epsilon' < 1$. If $v$ is a white vertex, we place $w_1$ and $w_2$ directly above and below $v$ at a distance of $\delta'$ from $v$, where $\sqrt{1 - \delta^2} < \delta' < 1$ and $\delta' > \epsilon'$. Observe that $\epsilon',\delta' > \sqrt{\sqrt{2} - \frac{1}{2}}$. Thus any two upper pseudo nodes have a distance of at least $\sqrt{(\delta'-\epsilon')^2+1}>1$ from one another (and similarly for lower pseudo nodes) and the distance between an upper pseudo node and a lower pseudo node is at least $2  \sqrt{\sqrt{2} - \frac{1}{2}} > 1$. Furthermore, the distance between a black upper (lower) pseudo node and a black lateral pseudo node is at least $\sqrt{{\epsilon'}^2 + \epsilon^2} > 1$, and similarly, the distance between a white upper (lower) pseudo node and a white lateral pseudo node is at least $\sqrt{{\delta'}^2 + \delta^2} > 1$.\\~\\
Let $U$ be the set of pseudo nodes. Let $n = |V|$ and let $n_W$ and $n_B$ be the number of black and white lateral pseudo nodes respectively. Let ${n_W}'$ be the number of white upper and lower pseudo nodes, and let ${n_B}'$ be the number of black upper and lower pseudo nodes. We now show that $G$ has a Hamiltonian path if and only if there is a spanning tree $S$ on the set of points $V \cup U$ with a maximum degree of 5, such that $w(S) \leq n-1 + \delta n_B + 2 \delta' {n_B}'  + \epsilon n_W + 2  \epsilon' {n_W}'$, where $w(S)$ denotes the sum of the lengths of the edges in $S$. If $G$ has Hamiltonian path, then for each vertex $v \in V$, we can add the edges $(v,u),(v,w_1),(v,w_2)$ to the path, where $u,w_1,w_2$ are the lateral pseudo node, upper pseudo node, and lower pseudo node of $v$ respectively, to give a spanning tree $S$ with a maximum degree of 5, and with weight $w(S) = n-1 + \delta n_B + 2 \delta' {n_B}'  + \epsilon n_W + 2  \epsilon' {n_W}'$. Conversely, suppose that we have a spanning tree $T$ on the points $V \cup U$ with a maximum degree of 5, and with weight $w(S) \leq n-1 + \delta n_B + 2 \delta' {n_B}'  + \epsilon n_W + 2  \epsilon' {n_W}'$. Note that the only edges whose lengths are strictly less than 1 are edges between the vertices in $V$ and their respective pseudo nodes in $U$. Since the only edges of unit length are those between adjacent vertices in $G$, and since all other edges have lengths strictly greater than 1, we can conclude that $T$ consists only of the edges from $E$ and the edges $(v,u),(v,w_1),(v,w_2)$ between the points $v \in V$ and their respective pseudo nodes in $U$, since the total weight of these edges is $n-1 + \delta n_B + 2 \delta' {n_B}'  + \epsilon n_W + 2  \epsilon' {n_W}'$ and these are the $(n+n_B+ {n_B}'+n_W+{n_W}' - 1)$ shortest edges. In this case, we can remove the edges $(v,u),(v,w_1),(v,w_2)$ from $T$ and we are left with a Hamiltonian path in $G$. Thus, there is a spanning tree $S$ on the set of points $V \cup U$ with a maximum degree of 5, where $w(S) \leq n-1 + \delta n_B + 2 \delta' {n_B}'  + \epsilon n_W + 2  \epsilon' {n_W}'$ if and only if $G$ has a Hamiltonian path. This completes the proof.
\end{proof}

We can now use this result to establish complexity results for the bottleneck version of the problem. It has already been shown that the $2$-E2MBST and $3$-E2MBST problems are NP-complete \cite{andersen2016minimum}, hence it follows immediately that the $2$-E3MBST and $3$-E3MBST problems are NP-complete as well. The proof of the previous theorem implies the following results.

\begin{theorem} \label{thm:deg5NPMBST}
 The 5-E3MBST problem is NP-complete.   
\end{theorem}
\begin{proof}
Given a grid graph $G$, perform the same construction with pseudo nodes as in the proof of Theorem~\ref{thm:deg5NPMST} to obtain the point set $P$. We claim that $P$ has a spanning tree of maximum degree 5 whose longest edge is of unit length if and only $G$ has a Hamiltonian path. If $G$ has a Hamiltonian path, then this path is a tree with a bottleneck value of 1. We can then add the edges $(v,u),(v,w_1),(v,w_2)$ between nodes and their respective pseudo nodes to extend the tree to a spanning tree. This will produce a spanning tree with a maximum degree of 5 and a bottleneck value of 1, since none of the edges between nodes and their respective pseudo nodes have a length greater than 1. Now suppose that $P$ has a spanning tree $T$ with maximum degree 5 and unit bottleneck value. Note that the only edges whose lengths are strictly less than 1 are those between nodes and their respective pseudo nodes. All edges of unit length are between adjacent nodes in $G$. Hence we can conclude that $T$ consists of edges of $G$ and edges $(v,u),(v,w_1),(v,w_2)$ between nodes and their respective pseudo nodes. If we remove the edges $(v,u),(v,w_1),(v,w_2)$ for each node $v$, then the resultant tree we are left with is a Hamiltonian path for $G$. Hence the result is proven.                     
\end{proof}

\begin{theorem} \label{thm:deg4NPMBST}
 The 4-E3MBST problem is NP-complete.   
\end{theorem}
\begin{proof}
This proof is the same as the proof of Theorem~\ref{thm:deg5NPMBST}, only we add 2 pseudo nodes to each node of $G$ instead of 3; a lateral pseudo node and an upper pseudo node. Hence, adding the edges $(v,u),(v,w_1)$ to a Hamiltonian path of $G$ produces a spanning tree of maximum degree 4 instead of 5.
\end{proof}

As was the case for the complexity result for the $3$-E2MBST in \citep{andersen2016minimum}, we can also establish inapproximability results as a corollary.

\begin{corollary} \label{cor:deg4NPMBST}
For $\delta = 4,5$, the $\delta$-E3MBST problem cannot be approximated in polynomial time within a constant factor less than $1.0009$ unless P=NP.  
\end{corollary}
\begin{proof}
The arguments for the case in which $\delta =4$ are the same for when $\delta = 5$ (we simply remove the lower pseudo nodes when $\delta = 4$), so we will assume for simplicity that $\delta = 5$. Suppose we had an algorithm $A$ that could approximate the problem within a constant factor $c$ in polynomial time. This means that if we were given an instance of the problem whose optimal bottleneck length is $L$, then if we performed $A$ on the same instance, it would yield a solution with a bottleneck value no worse than $cL$.
If we applied $A$ on set of points $V \cup U$ that we constructed in the proof of Theorem~\ref{thm:deg5NPMBST}, with feasible values for $\epsilon,\delta,\epsilon',\delta'$, then the algorithm would return a value $k \geq 1$. Let $\epsilon = 0.292249$, $\delta = \frac{\epsilon}{2} = 0.1461245$, $\epsilon' = \frac{2-\epsilon^2}{2} \approx 0.95729526$, and $\delta' = 0.9999999$, where $\epsilon, \delta, \epsilon'$ and $\delta'$ were chosen so as to make certain distances between pairs of pseudo nodes approximately equal.\\~\\
Note that the distance between any two black lateral pseudo nodes is at least
\begin{align*}
\sqrt{2}(1- \epsilon)
\approx 1.00091,
\end{align*} 
the distance between any white lateral pseudo node and any black lateral pseudo node is at least
\begin{align*}
\sqrt{(\epsilon-\delta)^2+1}
\approx 1.01062,
\end{align*}
and the distance between a white lateral pseudo node and an original black node is at least
\begin{align*}
\sqrt{1+\delta^2}
\approx 1.01062.
\end{align*}
In addition, the distance between an upper (lower) black pseudo node and a lateral black pseudo node is at least
\begin{align*}
\sqrt{{\epsilon'}^2+\epsilon^2}
\approx 1.00091,
\end{align*}
the distance between an upper (lower) white pseudo node and a lateral white pseudo node is at least
\begin{align*}
\sqrt{{\delta'}^2+\delta^2}
\approx  1.01062,
\end{align*}
and the distance between an upper (lower) black pseudo node and an upper (lower) white pseudo node is at least
\begin{align*}
\sqrt{(\delta'-\epsilon')^2+1}
\approx  1.00091,
\end{align*}

Hence, if $k< 1.0009$, then $k$ must equal $1 $ since no pair of points in $V \cup U$ have a distance that lies in the open interval $(1,1.0009)$. Hence, if $c < 1.0009$, we could use our approximation algorithm to determine if $G$ has a Hamiltonian path in polynomial time. Thus if P $\neq$ NP, we must have $c \geq 1.0009$.
\end{proof}

We also consider the complexity of the rectilinear version of the problem. As in the case of the Euclidean version, it has already been shown that the $2$-R2MBST and $3$-R2MBST problems are NP-complete \cite{andersen2016minimum}, hence it follows that the $2$-R3MBST and $3$-R3MBST problems are NP-complete. To establish complexity results for the cases in which $\delta$ is equal to $4$ and $5$, we can use the same approach as in the Euclidean versions.

\begin{theorem}
The $\delta$-R3MBST problem is NP-complete for $\delta \in \{4,5\}$. Furthermore, the problem cannot be approximated in polynomial time within a constant factor less than  $\frac{6}{5} = 1.2$ unless P=NP.
\end{theorem}
\begin{proof}
Once again, we will just consider the $\delta = 5$ case since the arguments for the $\delta = 4$ case are identical, with the exception of not using lower pseudo nodes. Given a grid graph $G$, perform the same construction with pseudo nodes as in the proof of Theorem~\ref{thm:deg5NPMBST} to obtain the point set $P$. However, this time we will set $\epsilon = \frac{2}{5}$, $\delta = \frac{1}{5}$, $\epsilon' = \frac{4}{5}$, $\delta' = 1 - \Delta$, where $\Delta > 0$ is some arbitrarily small number, no bigger than say 0.01.\\~\\
Now note that the distance between any two black lateral pseudo nodes is at least
\begin{align*}
2(1- \epsilon)
= \frac{6}{5},
\end{align*} 
the distance between any white lateral pseudo node and any black lateral pseudo node is at least
\begin{align*}
(\epsilon-\delta)+1
= \frac{6}{5},
\end{align*}
and the distance between a white lateral pseudo node and an original black node is at least
\begin{align*}
1+\delta
= \frac{6}{5}.
\end{align*}
In addition, the distance between an upper (lower) black pseudo node and a lateral black pseudo node is at least
\begin{align*}
\epsilon'+\epsilon
= \frac{6}{5},
\end{align*}
the distance between an upper (lower) white pseudo node and a lateral white pseudo node is at least
\begin{align*}
\delta'+\delta
= \frac{6}{5}-\Delta,
\end{align*}
and the distance between an upper (lower) black pseudo node and an upper (lower) white pseudo node is at least
\begin{align*}
(\delta'-\epsilon')+1
= \frac{6}{5}-\Delta,
\end{align*}

Hence, we can use the arguments of Theorem~\ref{thm:deg5NPMBST} to conclude that the problem is NP-complete. Furthermore, suppose we had an algorithm $A$ that could approximate the problem within a constant factor $c$ in polynomial time, and that $A$ returned a value of $k$ when given the set of constructed points as input. If $k< \frac{6}{5}-\Delta$, then $k$ must equal $1 $ since no pair of points in $V \cup U$ have a distance that lies in the open interval $(1,\frac{6}{5}-\Delta)$. Hence, if $c < \frac{6}{5}$, then for a sufficiently small choice of $\Delta$, we could use our approximation algorithm to determine if $G$ has a Hamiltonian path in polynomial time. Thus if P $\neq$ NP, we must have $c \geq \frac{6}{5}$.

\end{proof}


Since grid graphs have a maximum degree of 6, we cannot extend this pseudo node approach to higher values of $\delta$ without significant alterations to the structure of the proof. It is possible that the approach of Francke and Hoffman \cite{francke2009euclidean} may be extended to three dimensions, however it is an open problem of how to convert their min-sum result into a bottleneck result.

\section{Approximation Algorithms for the $\delta$-E3MBST Problem}
In this section we describe ways of extending Khuller, Raghavachari, and Young's algorithm for the 3-EMST \cite{khuller1996low}, which we will refer to as the KRY algorithm, to algorithms for the $\delta$-E3MBST problem. 
The KRY algorithm, presented here as Algorithm~\ref{alg:khuller}, starts with a rooted MST $T$ for the input point set in the Euclidean plane, which it proceeds to process recursively, performing local edge swaps whenever the current root node has a degree exceeding $3$. The edge swaps themselves are applied in the following manner. If $v$ is the current root of $T$ with children $v_1,v_2, \dots, v_k$, then the edges $(v,v_2), \dots, (v,v_k)$ are replaced by a path through the vertices $v_1,v_2, \dots, v_k$. The algorithm is then applied recursively to each of the subtrees rooted at $v_1,v_2, \dots, v_k$ in turn, which we denote by  $T_{v_1},T_{v_2}, \dots, T_{v_k}$ respectively. See Figure~\ref{fig:ChanPaperKhullerFig} for an illustration. The edges from the root node $v$ to the children are replaced by a path through the children starting at $v$, however this path may not be unique. To adapt this algorithm to the MBST problem, we will assume that the paths chosen are those that minimise the length of the longest edge in the path.\\

\begin{algorithm}[htb]
\caption{: KRY} \label{alg:khuller}
\begin{tabbing}
  ....\=....\=....\=....\=................... \kill \\ [-2ex]
\textbf{Input:} A rooted tree $T$ over a point set $P$ with root $v$ and a partially built\\ solution $T^*$.\\
\\
\textbf{if} $v$ has at least one child\\
\> Let the children of $v$ be $v_1, \dots, v_k$, where $k$ is the number of children of $v$,\\
\> such that $\max_{i \in [1, k-1]} w(v_i, v_{i+1})$ is minimum if $k>1$.\\
\> Add the edge $(v,v_1)$ to $T^*$\\
\> Perform Algorithm~\ref{alg:khuller} with $T:=T_{v_1}$ as input.\\
\> \textbf{if} $k \geq 2$\\
\>\> \textbf{for} $i=1,\dots, k-1$\\
\>\>\> Add the edge $(v_i,v_{i+1})$ to $T^*$\\
\>\>\> Perform Algorithm~\ref{alg:khuller} with $T:=T_{v_{i+1}}$ as input.\\
\\
\textbf{Output:} $T^*$.
\end{tabbing}
\end{algorithm}

\begin{figure} 
        \centering
        \includegraphics[scale=0.7]{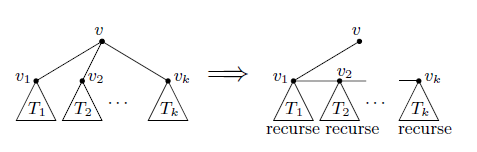}
        \caption{The illustration of the Khuller, Raghavachari and Young algorithm \cite{chan2003euclidean}.}
\label{fig:ChanPaperKhullerFig}        
\end{figure}

It has been shown in \cite{andersen2016minimum} that Algorithm~\ref{alg:khuller} is a 2-factor approximation algorithm for the 3-MBST problem in arbitrary metric spaces, hence we can conclude that it is a 2-factor approximation algorithm for the $3$-E3MBST Problem. One can see that this bound is tight when given an MST such as that of Figure~\ref{fig:KRYwcase2}, which is an MST with a maximum degree of 4 in which all edges are of equal length. In this MST, no matter which vertex is chosen as the initial root vertex, the algorithm will arrive at a local subtree with two children in which the root vertex $v$ and its children $v_1$ and $v_2$ are collinear, with $v$ between $v_1$ and $v_2$. In this situation, the algorithm will swap in the edge $(v_1,v_2)$ whose length is equal to the sum of the lengths of $(v,v_1)$ and $(v,v_2)$, i.e., twice the bottleneck length.\\

\begin{figure} 
        \centering
        \includegraphics[scale=0.4]{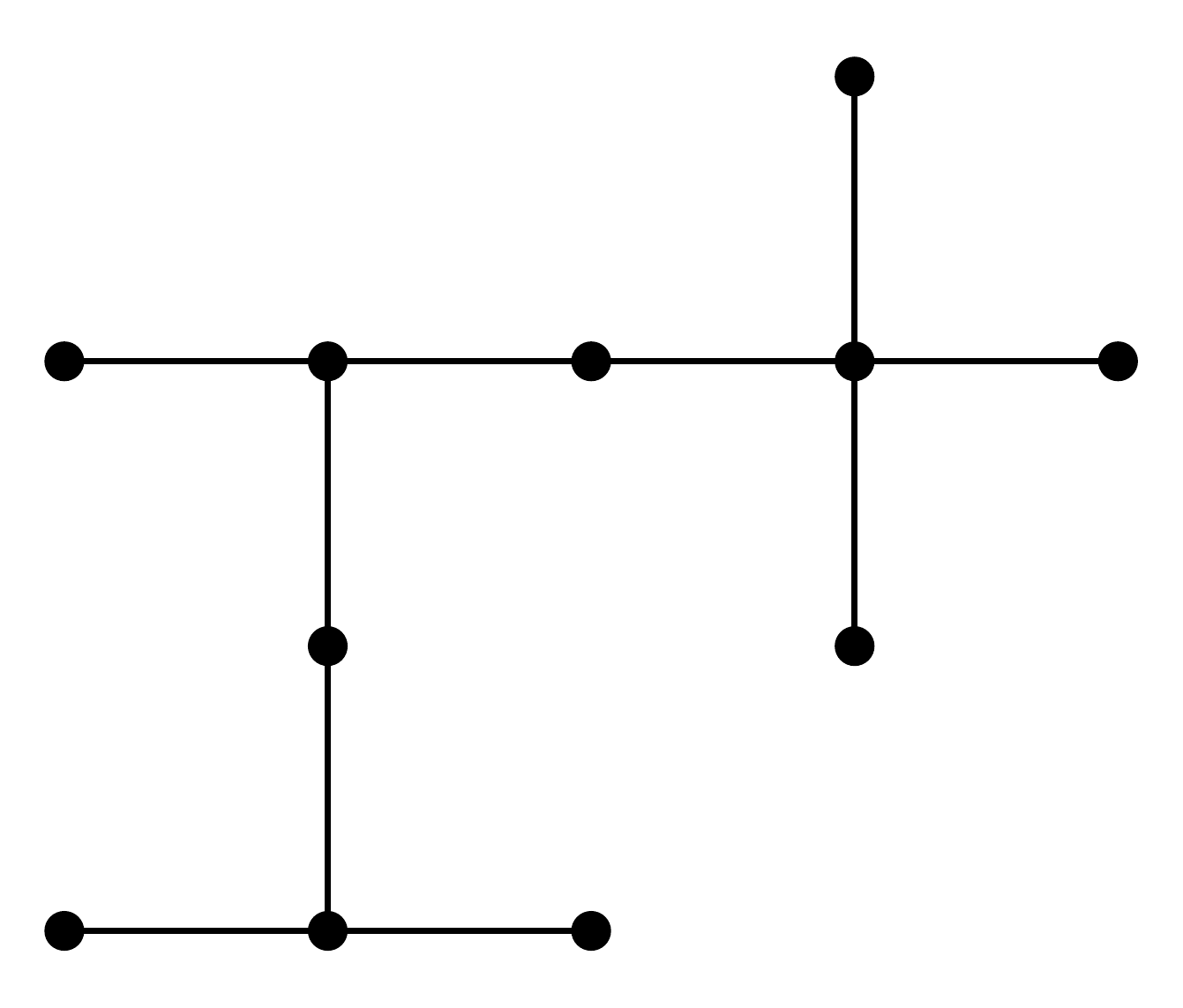}
        \caption{An example of an MST in which the bottleneck value of the KRY algorithm's output will be twice that of the MST no matter which vertex is chosen as the initial root vertex. Distances between adjacent nodes in the MST are assumed to be equal.}
\label{fig:KRYwcase2}        
\end{figure}
As it stands, this algorithm is only suited to the $\delta$-MBST problem when $\delta = 3$ since it cannot guarantee a resultant tree whose maximum degree is at most 2, and it will always produce a tree whose maximum degree is at most 3, even if our degree bound allows for larger degrees.\\~\\ 
There are multiple ways one could extend the approach of the KRY algorithm to higher values of $\delta$. The simplest of these ways is merely to change the minimum number of children a given root vertex must have before an edge swap must take place, whilst maintaining the current form of the edge swaps, i.e., creating a path through the children. Currently, a root vertex $v$ must have at least two children, $v_1$ and $v_2$, before 
an edge swap takes place (i.e., $(v_1,v_2)$ replaces either $(v,v_1)$ or $(v,v_2)$). In this way, the resultant tree is guaranteed to have a maximum degree of no more than 3. If we relax this condition so that the edge swaps are only performed when the current root node has a degree of $\delta - 1$ or higher, then the tree output by the algorithm will have a maximum degree of no more than $\delta$. We refer to such an algorithm as the \textit{naive} KRY algorithm for the $\delta$-MBST, or  $\delta$-NKRY for short. Here the term naive refers to the choice of always using a path through the children as the edge swaps, which results in the root vertex having a degree no more than 3 as per the original KRY algorithm, even though we are allowed to have vertices of higher degree. In this way, the usual KRY algorithm is the $3$-NKRY algorithm. We describe the $\delta$-NKRY algorithm in Algorithm~\ref{alg:Nkhuller}.\\

\begin{algorithm}[htb]
\caption{: $\delta$-NKRY} \label{alg:Nkhuller}
\begin{tabbing}
  ....\=....\=....\=....\=................... \kill \\ [-2ex]
\textbf{Input:} A rooted tree $T$ over a point set $P$ with root $v$ and a partially built\\ solution $T^*$.\\
\\
\textbf{if} $v$ has at least one child\\
\> Let the children of $v$ be $v_1, \dots, v_k$, where $k$ is the number of children of $v$,\\
\> such that $\max_{i \in [1, k-1]} w(v_i, v_{i+1})$ is minimum if $k>1$.\\
\> Add the edge $(v,v_1)$ to $T^*$\\
\> Perform Algorithm~\ref{alg:Nkhuller} with $T:=T_{v_1}$ as input.\\
\> \textbf{if} $k \geq \delta-1$\\
\>\> \textbf{for} $i=1,\dots, k-1$\\
\>\>\> Add the edge $(v_i,v_{i+1})$ to $T^*$\\
\>\>\> Perform Algorithm~\ref{alg:Nkhuller} with $T:=T_{v_{i+1}}$ as input.\\
\\
\textbf{Output:} $T^*$.
\end{tabbing}
\end{algorithm}

\subsection{Performance Ratio of the $\delta$-KRY Algorithm}

In order to analyse the worst case performance ratio of the $\delta$-NKRY algorithm for the $\delta$-E3MBST problem, we wish to characterise the instances in $E^3$ for which the algorithm yields the worst performance. Since the $\delta$-NKRY algorithm uses local edge swaps and does not swap out edges that were swapped in at an earlier stage of the algorithm, we need only characterise instances which yield the worst performance in a single iteration of the algorithm, where an iteration involves edge swaps incident to a given root node. In other words, we need only find a \textit{star} which gives the worst performance for the algorithm, where a star is an MST that contains only a single root vertex and its children.\\~\\ 
Let $P$ be a set of points in $\mathbb{R}^d$. Let $T$ be an MST of $P$ rooted at a vertex $v$, and let $v_1, \dots v_k$ be the children of $v$, where $k \in \mathbb{N}$. In this way, we can think of $v$ and its children as a star with $v$ as the centre vertex. Due to the geometry of the MST in $\mathbb{R}^d$, the arrangement of points in such a star can be seen to have a particular structure, which we will illustrate below.\\ ~\\
Consider the convex hull of the vertices of the star. It was shown in \cite{khuller1996low} that when $d = 2$, the convex hull of $\{v,v_1, \dots, v_k\}$ would contain each of the children $\{v_1, \dots, v_k\}$ on its boundary ($v$ may or may not be an interior point of the convex hull depending on the arrangement of the star). This fact is a consequence of the following result given in \cite{khuller1996low}.

\begin{lemma}\label{lem:krycor}
Let $AB$ and $BC$ be two edges incident to a point $B$ in an MST of a set of points in $\mathbb{R}^d$. Then
\begin{itemize}
\item $\angle ABC \geq 60^\circ$,
\item $\angle BAC, \angle BCA \leq 90^\circ$.
\end{itemize}
\end{lemma}

We now use the previous lemma to extend the result about convex hulls to higher dimensions.

\begin{lemma}
 Let $T$ be an MST of a set of points in $\mathbb{R}^d$. Let $v$ be the root of $T$ and let $v_1, \dots v_k$ be the children of $v$. Then the convex hull of $\{v,v_1, \dots, v_k\}$ contains every point in $\{v_1, \dots, v_k\}$ on its boundary.
\end{lemma}
\begin{proof}
The statement is trivially true when $k \leq 2$, so we will assume that $k \geq 3$. In order to treat the points $v,v_1, \dots, v_k$ as vectors, we will also assume that the points have been translated so that $v$ lies at the origin in $\mathbb{R}^d$. Let $C$ be the convex hull of $\{v,v_1, \dots, v_k\}$. Suppose there exists a point $x$ in $\{v_1, \dots, v_k\}$ that lies in the interior of $C$. Let $x'$ be the unique point on the boundary of $C$ that is obtained by scaling $x$ in the positive direction. Clearly $x' \notin \{v_1, \dots, v_k\}$ by Lemma~\ref{lem:krycor} considering the edges $(v,x)$ and $(v,x')$. Let $v_i, v_j$ be two points in $\{v,v_1, \dots, v_k\}$ such that $v_i, v_j$ and $x'$ all lie on the same facet of $C$. Let $\theta_1 =\angle (v,x,v_i)$ and  $\theta_2 =\angle (v,x,v_j)$.\\~\\
Let $P$ be the unique plane which passes through $v_i, v_j$ and $v$. Let $y$ be the projection of $x$ onto $P$, an let $y'$ be the projection of $x'$ onto $P$. Hence $v,v_i,v_j,y,y'$ are coplanar, and $v_i,y',v_j$ are collinear. The polygon on $P$ given by $(v,v_i,y',v_j)$ will be a triangle, and the polygon on $P$ defined $(v,v_i,y,v_j)$ will be non-convex (see Figure~\ref{fig:convex1}). In particular, the interior angle of $y$ in the non-convex polygon will be greater than $180^\circ$. Since this angle is also the sum of the angles given by $\omega_1 = \angle (v,y,v_i)$ and  $\omega_2 = \angle (v,y,v_j)$, this implies that one of the angles $\omega_1, \omega_2$ is greater than or equal to $90^\circ$. However, we have that $\theta_1 \geq \omega_1$ and  $\theta_2 \geq \omega_2$ since $y$ is a projection of $x$ onto the plane defined by $v_i$ and $v_j$. Hence one of the angles $\theta_1, \theta_2$ is greater than or equal to $90^\circ$, which contradicts Lemma~\ref{lem:krycor}.
\end{proof}


\begin{figure}[ht!]
\centering
 \includegraphics[scale =1]{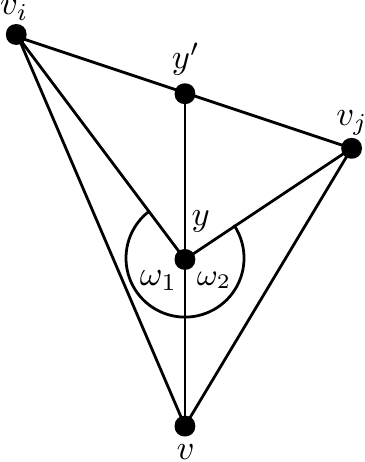}
 \caption{An example of the points and angles described by Lemma~\ref{lem:krycor}.}
 \label{fig:convex1} 
\end{figure}

To simplify the analyses of our algorithms, we will only consider the cases in which the children of a given root vertex $v$ are all equidistant from $v$. We are able to restrict our cases to this specification since the ratios obtained under this assumption are an upper bound to the ratios that would be obtained otherwise. This is due to the following lemma.

\begin{lemma}\label{lem:sidez}
Let $x$,$y$ and $z$ be the lengths of edges of a triangle with $x \leq y$, and the angle between the sides of $x$ and $y$ is $\theta \geq 60^\circ$. If $\theta$ and $y$ are fixed, and if the rule $x \leq y$ must be maintained, then $z$ is maximised when $x = y$.
\end{lemma}
\begin{proof}
Using the cosine rule,
\[z^2 = x^2 +y^2 - 2xy \cos(\theta).\]
Using calculus, we can see that the function $f(x) = z^2$ is convex so the maximum value of the function will occur at either $x = 0$ or $x =y$, assuming $x \in [0,y]$. We have that $f(0) = y^2$ and $f(y) = 2y^2(1-\cos(\theta))$. If we have that $f(0) > f(y)$, then
\[y^2 > 2y^2(1-\cos(\theta)),\]
\[ \Rightarrow 1 > 2-2\cos(\theta),\]
\[ \Rightarrow \cos(\theta) > \frac{1}{2},\]
which is a contradiction since $60^\circ \leq \theta \leq 180^\circ$. Hence $f(x)$ is maximised when $x = y$.
\end{proof}

Let $V$ be a set of vertices in a metric space and let $v \in V$. We define the \textit{optimal bottleneck TSP-path} through $V$ starting at $v$ as the path $P$ such that $P$ visits all vertices of $V$, $v$ is an endpoint of $P$, and the length of the longest edge of $P$ is minimum, and we let $b_v(V)$ denote the length of the longest edge in $P$. Also for simplicity, we will say that a star is \textit{valid} if it does not contain a pair of children that form an angle less than $60^\circ$ with the root vertex. In order to justify our assumption that the children are equidistant from the root vertex $v$ in the worst case, we apply Lemma~\ref{lem:sidez} in the following theorem.

\begin{theorem} \label{thm:equidistant}
Let $S$ be a valid star with root vertex $v$ and children $C = \{v_1, \dots, v_k\}$. Let $l_{\max}$ be the maximum distance $d(v,v_i)$ for $i \in \{1, \dots, k\}$ and let $P$ be the optimal bottleneck TSP-path through $\{v,v_1, \dots, v_k\}$ that starts at $v$. Then there exists a valid star $S'$ with root vertex $v$ and children $C' = \{v'_1, \dots, v'_k \}$ such that $d(v,v'_i) = l_{\max}$ for $i \in \{1, \dots, k\}$, and $b_v(C') \geq b_v(C)$.

\end{theorem} 

\begin{proof}
Let $v_1, \dots, v_k$ be the children of $v$, where $k \geq 2$, and let $P$ be the optimal bottleneck TSP-path starting as defined in the theorem. If not all children are equidistant from $v$, then there exists one or more children of distance $l_{\max}$ from $v$, and there exists one or more children whose distances from $v$ are strictly less than $l_{\max}$. Let $V$ be the set of children whose radial distances from $v$ are $l_{\max}$. Let $l$ be maximum radial distance of any child such that $l < l_{\max}$, i.e., $l$ is the second largest radial distance, and let $U$ be the set of children whose radial distances are $l$. Let $W$ be the set of remaining children that are in neither $V$ nor $U$. Simultaneously scale all radial distances of all the vertices in $U$ so that their radial distances become $l_{\max}$ and let $P'$ be the optimal bottleneck TSP-path for this modified set of children. We claim that the longest edge in $P'$ is no shorter than the longest edge in $P$. To see this, first note that for any pair of vertices $(w,u)$ where $w \in W$ and $u \in U$, the radial distances of the vertices in $U$ were strictly greater than those in $W$ before the scaling, so by the cosine rule, $d(w,u)$ will increase after the scaling. Also, for any pair of vertices $(v,u)$ where $v \in V$ and $u \in U$, Lemma~\ref{lem:sidez} implies that $d(v,u)$ will not decrease after the scaling. Since the distances between all other pairs of vertices will either remain the same or increase, we can conclude that it is impossible for $P'$ to have a bottleneck length strictly less than that of $P$ without contradicting the optimality of $P$. By repeating this process, we can conclude that if we scale the radial distances of all children so that they were of distance $l_{\max}$ from $v$, then the optimal bottleneck TSP-path $\hat{P}$ through $v$ and all the children in this scaled set, where $\hat{P}$ starts at $v$, would have a bottleneck value that was greater than or equal to that of $P$.
\end{proof}



Therorem~\ref{thm:equidistant} can be used to restrict the types of instances that need to be considered when attempting to find an instance for which the algorithm yields a large performance ratio. Given a value for $\delta$, in order to obtain a lower bound for the approximation factor for the $\delta$-NKRY algorithm in three dimensions, we would need only find a way of arranging $\delta - 1$ points on the surface of a unit sphere, where no pair of points produces an angle less than $60^\circ$ with the centre of the sphere, such that the value of the solution to the Euclidean bottleneck TSP-path problem for the points is maximised. The bottleneck value of this path will give the lower bound for the worst case performance ratio. In this way, we will attempt to produce stars that yield the worst possible performance ratio for the algorithm, where the center of the sphere is the root vertex and the children are the points on the surface of the sphere. For the remainder of this section, we aim to find these worst case arrangements of $k$ points on the sphere for the various values of $k \in \{2, \dots, 10 \}$ as these would provide lower bounds for the worst case performance of the $k+1$-NKRY algorithm.

\subsection{Worst Case Ratio for the $3$-NKRY Algorithm (2 Children)}
As mentioned previously, it is known that the $3$-NKRY algorithm is a 2-factor approximation algorithm in any metric space. The arrangement of points on the surface of a unit sphere that produces this result is the arrangement in which the two children are polar opposite, as seen in Figure~\ref{fig:points2}. Thus, given a root vertex at the centre of a unit sphere with two children
at opposite poles of the sphere, the KRY algorithm will replace one of the unit radii between the root and a child with the diameter length edge between the two children, hence we will obtain a performance ratio of 2 for this star. Note that it is not possible to position two points on the surface of a unit sphere so that the distance between the points is larger than 2.
\begin{figure}[ht!]
\centering
 \includegraphics[scale =0.7]{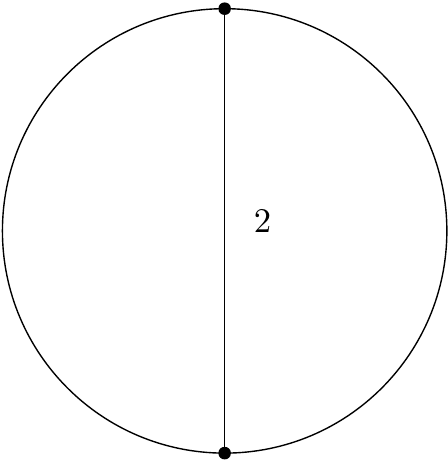}
 \caption{The worst case arrangement of 2 points on a unit sphere.}
 \label{fig:points2} 
\end{figure}

\subsection{Worst Case Ratio for the $4$-NKRY Algorithm (3 Children)}
To analyse this algorithm, we aim to find an arrangement of 3 points on the surface of the unit sphere such that minimum bottleneck value of all paths through these 3 vertices is maximised, where no pair of vertices produce an angle less than $60^\circ$ with the centre (i.e., every vertex has a distance of least 1 from every other vertex). We find such an arrangement in the following theorem.

\begin{theorem} \label{thm:3childarrangement}
The $4$-NKRY algorithm has an approximation factor of at least 1.931 for the $4$-E3MBST problem.
\end{theorem}

\begin{proof}
Suppose we had an arrangement of 3 points as described above. Let $P$ be a path through the vertices that gives the minimum bottleneck value. Let $v_1$ be a point on this sphere that is incident to a longest edge in $P$. We will assume without loss of generality that $v_1$ is on the north pole of the sphere. Let $v_2$ and $v_3$ be the remaining vertices. Then the bottleneck edge of $P$ is the smaller of the two edges $(v_1,v_2)$ and $(v_1,v_3)$. Observe that $v_2$ and $v_3$ must be equidistant from $v_1$, since if this were not the case, we could rotate $v_2$ and $v_3$ around the surface of the sphere, in the direction of the arc between $v_2$ and $v_3$, so that the smaller of $d(v_1,v_2)$ and $d(v_1,v_3)$ is increased whilst keeping the length of $d(v_2,v_3)$ fixed, which contradicts the maximality of the arrangement. Hence we conclude that $d(v_1,v_2) = d(v_1,v_3)$. These lengths are maximised when $v_1, v_2$ and $v_3$ lie on a common circle of radius 1 which will contain the south pole. Hence $d(v_1,v_2)$ and $d(v_1,v_3)$ are maximised when $v_2$ and $v_3$ are as close together as possible on the southern part the circle, i.e., $d(v_2,v_3) = 1$. Figure~\ref{fig:points3} shows the worst case arrangement for which one can use trigonometry to establish that the bottleneck length of the path is $\frac{\sin 75^\circ}{\sin 30^\circ} = \frac{\sqrt{2} + \sqrt{6}}{2} \approx 1.931$. Hence we conclude that $4$-NKRY algorithm has an approximation factor of at least 1.931 for the $4$-E3MBST problem.

\begin{figure}[ht!]
\centering
 \includegraphics[scale =0.7]{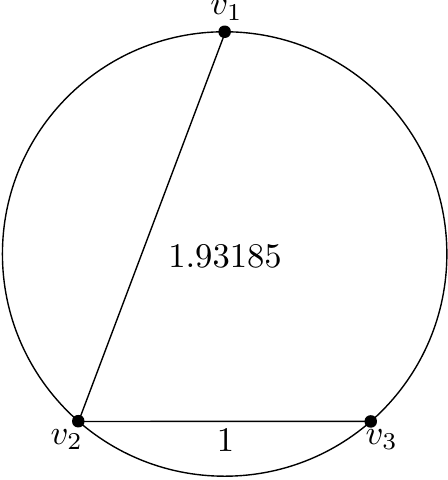}
 \caption{The worst case arrangement of 3 points on a unit sphere.}
 \label{fig:points3} 
\end{figure}

\end{proof}

\subsection{Worst Case Ratio for the $5$-NKRY Algorithm (4 Children)}
Analysing the worst case arrangement of 4 points is slightly more difficult than with 3 points as there are more cases to consider. The details are presented in the following theorem.

\begin{theorem} \label{thm:4childarrangement}
The $5$-NKRY algorithm has an approximation factor of at least 1.906 for the $5$-E3MBST problem.
\end{theorem}

\begin{proof}
As in the proof of Theorem~\ref{thm:3childarrangement}, we assume we have a set of points that are in an arrangement such that the bottleneck value of the minimum bottleneck path $P$ through the points is maximised. Let $v_1$ be a vertex on the north pole incident to the bottleneck edge, and let $v_2,v_3,v_4$ be the remaining vertices. First, consider the case where $v_1$ is an endpoint of $P$. Intuition would suggest that $v_2,v_3$ and $v_4$ are on the surface of the southern hemisphere as far away from $v_1$ as possible. Consider the plane containing $v_2,v_3$ and $v_4$. If $v_1$ is also on this plane, then $v_1,v_2,v_3$ and $v_4$ are co-circular and we can use similar arguments to those of the previous section to justify the possible arrangement of points as in Figure~\ref{fig:points4a} with a minimum bottleneck value of $\sqrt{3}$, where $v_3$ is on the south pole with a distance of 1 from $v_1$ and $v_2$ which are either side of $v_3$.\\

\begin{figure}[ht!]
\centering
 \includegraphics[scale =0.7]{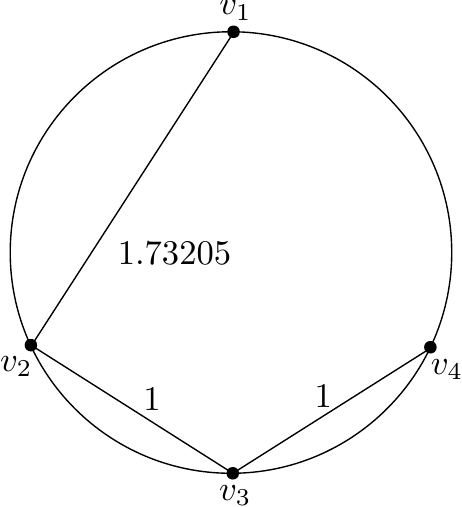}
 \caption{An arrangement of 4 points on a unit sphere such that the minimum bottleneck path length is $\sqrt{3}$.}
 \label{fig:points4a} 
\end{figure} 
If $v_1$ is not co-planar with $v_2,v_3$ and $v_4$, then consider the geodesic triangle of $v_2,v_3$ and $v_4$ formed by the geodesics between the vertices on the surface of the sphere. If the south pole does not lie in the interior or boundary of the triangle, then we could rotate the triangle around the sphere towards the south pole to increase the distance between $v_1$ and all other points, contradicting the maximality of the arrangement. If $v_2,v_3$ and $v_4$ are co-linear on the surface of the sphere, then we can use similar arguments as in the previous section to arrive at the possible arrangement in Figure~\ref{fig:points4a}. Assume that $v_2,v_3$ and $v_4$ are not co-linear so that the geodesic triangle of $v_2,v_3$ and $v_4$ on the surface of the sphere has non-zero area and assume that the south pole lies in the triangle's interior. If any single point of the triangle can be moved towards the centre whilst keeping the other points fixed, then such a move would not decrease the minimum bottleneck value of the arrangement, so we will assume that the vertices are too close together to allow such a movement. Hence there is vertex of the triangle that is of unit distance from the other two vertices in the triangle, and so the bottleneck edge will be between $v_1$ and its closest neighbour in the triangle. If the plane containing $v_2,v_3$ and $v_4$ is not parallel to the tangent plane to the sphere at $v_1$, then it would be possible to rotate the triangle around the surface of the sphere in such a way as to increase the distance between $v_1$ and its nearest neighbour, which is a contradiction. The only possibility for this arrangement is with the triangle being equilateral with the south pole at its centre, shown in Figure~\ref{fig:points4b}. This yields a bottleneck value of approximately 1.906.\\

\begin{figure}[ht!]
\centering
\begin{tikzpicture}
  
  \fill[fill=black] (0,2,0) circle (2pt) node[below right]{$v_1$};
  \fill[fill=black] (1.1547,-1.6330,0) circle (2pt) node[above right]{$v_2$};
  \fill[fill=black] (-0.5774,-1.6330,1.0000) circle (2pt) node[above]{$v_3$};
  \fill[fill=black] (-0.5774,-1.6330,-1.0000) circle (2pt)  node[above]{$v_4$};
  \draw (0,2,0) -- node[above right]{1.90604} (1.1547,-1.6330,0);
  \draw (1.1547,-1.6330,0) -- node[below]{1} (-0.5774,-1.6330,1.0000);
  \draw (-0.5774,-1.6330,1.0000) -- node[above]{1} (-0.5774,-1.6330,-1.0000);
  \draw[dashed] (1.1547,-1.6330,0) --  (-0.5774,-1.6330,-1.0000);
  \draw[dashed] (0,2,0) --  (-0.5774,-1.6330,-1.0000);
  \draw[dashed] (0,2,0) --  (-0.5774,-1.6330,1.0000);
\end{tikzpicture}
\caption{An arrangement of 4 points on a unit sphere such that the minimum bottleneck path length is 1.906.}
 \label{fig:points4b} 
\end{figure}
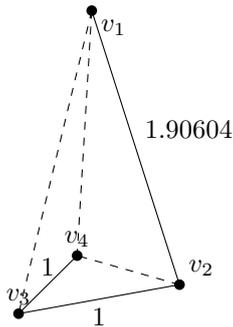 


Finally, we will consider the case that $v_1$ is not an endpoint of the path and that it is also incident to an edge of the path whose length is strictly less than the bottleneck value. Let the path $P$ be $(v_2,v_1,v_4,v_3)$. Due to symmetry, we also assume that the length of $(v_4,v_3)$ is less than the bottleneck value, otherwise we could use similar reasoning to the previous case. In this way we can partition the vertices into sets $\{v_1,v_2\}, \{v_3,v_4\}$ with the bottleneck edge going between sets. Since both $d(v_1,v_2)$ and $d(v_3,v_4)$ are less than the bottleneck value, the length of the bottleneck edge is the minimum distance between the two sets. Hence the two sets must be as far from one another as possible and the maximum distance will occur when the edges $(v_2,v_1)$ and $(v_4,v_3)$ are parallel and when $d(v_1,v_2) = d(v_3,v_4) = 1$. The optimal arrangement in this case is given by Figure~\ref{fig:points4c} and has a minimum bottleneck value of $\sqrt{3}$.\\

\begin{figure}[ht!]
\centering
 \includegraphics[scale =0.7]{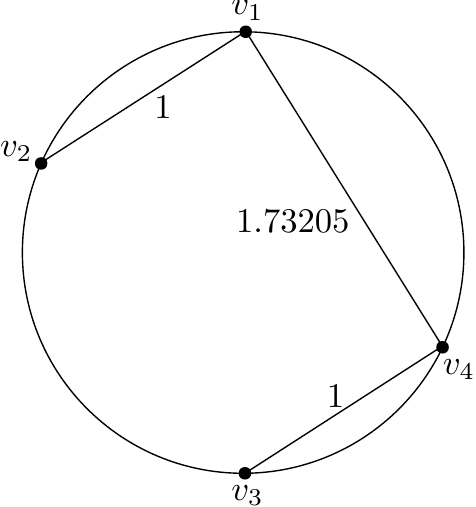}
 \caption{An alternate arrangement of 4 points on a unit sphere such that the minimum bottleneck path length is $\sqrt{3}$.}
 \label{fig:points4c} 
\end{figure}

Thus the arrangement that gives the maximum bottleneck value for the optimal bottleneck path is the one in which the vertices form a pyramid with an equilateral triangle base and we conclude that approximation factor of the $5$-NKRY algorithm for the $5$-E3MBST problem is at least 1.906.
\end{proof}

\subsection{Worst Case Ratio for the $\delta$-NKRY Algorithm for $\delta = 6,7,8$}
For the cases with 5 or more children, we will not use analytic techniques to establish proofs of the worst case ratios of the $\delta$-NKRY algorithm due to the complexity of the problem. Instead, we opt for computer search techniques in order find configurations of points experimentally. Whilst the numbers we obtain through our experiments are not confirmed analytically as absolute worst case ratios, they provide us with lower bounds for the worst case performances of the algorithm. For the cases with 6,7 and 8 children, our experimental results seem to suggest that the worst case configurations of points occur when the points are arranged in certain easily describable symmetric geometric patterns. We refer to these arrangements as \textit{representatives}.\\~\\
For 5 children, our representative was a square pyramid with the apex of the pyramid on the north pole of the sphere and the base of the pyramid having unit side length, where the vertices of the base are on the surface of the southern hemisphere (see Figure~\ref{fig:points5}). When we calculate the length of the bottleneck edge, which in this case is any the edge between the apex and a vertex of the base of the pyramid, we find that the length is approximately $1.8478$. Hence a lower bound for the worst-case approximation ratio of the $6$-NKRY Algorithm is $1.8478$.\\

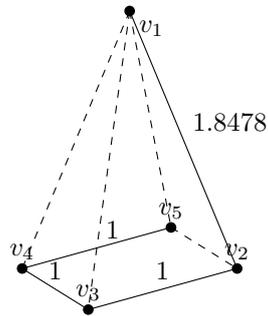
\begin{figure}[ht!]
\centering
\begin{tikzpicture}
  
  \fill[fill=black] (0,2,0) circle (2pt) node[below right]{$v_1$};
  \fill[fill=black] (1.4141,-1.4141,0) circle (2pt) node[above]{$v_2$};
  \fill[fill=black] (0,-1.4144,1.4141) circle (2pt) node[above]{$v_3$};
  \fill[fill=black] (-1.4141,-1.4144,0) circle (2pt)  node[above]{$v_4$};
  \fill[fill=black] (0,-1.4144,-1.4144) circle (2pt)  node[above]{$v_5$};
  \draw (0,2,0) -- node[above right]{1.8478} (1.4141,-1.4141,0);
  \draw (1.4141,-1.4141,0) -- node[above]{1} (0,-1.4144,1.4141);
  \draw (0,-1.4144,1.4141) -- node[above]{1} (-1.4141,-1.4144,0);
  \draw (0,-1.4144,-1.4144) -- node[above right]{1} (-1.4141,-1.4144,0);
  \draw[dashed] (0,-1.4144,-1.4144) --  (1.4141,-1.4141,0);
  \draw[dashed] (0,2,0) --  (0,-1.4144,1.4141);
  \draw[dashed] (0,2,0) --  (-1.4141,-1.4144,0);
  \draw[dashed] (0,2,0) --  (0,-1.4144,-1.4144);
\end{tikzpicture}
\caption{An arrangement of 5 points on a unit sphere such that the minimum bottleneck path length is 1.8478.}
 \label{fig:points5} 
\end{figure} 

Similarly, for 6 children, the computer search yielded a representative in the form of a pentagonal pyramid whose base has unit side length (see Figure~\ref{fig:points6}). In this case, the bottleneck edge has an approximate length of $1.7468$.\\

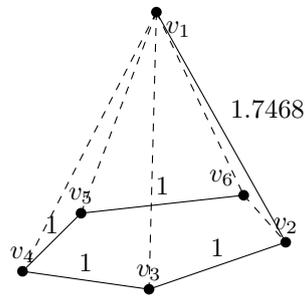
\begin{figure}[ht!]
\centering
\begin{tikzpicture}
  
  \fill[fill=black] (0,2,0) circle (2pt) node[below right]{$v_1$};
  \fill[fill=black] (1.7014,-1.0513,0) circle (2pt) node[above]{$v_2$};
  \fill[fill=black] (0.5258,-1.0513,1.6181) circle (2pt) node[above]{$v_3$};
  \fill[fill=black] (-1.3765,-1.0513,1.0001) circle (2pt)  node[above]{$v_4$};
  \fill[fill=black] (-1.3765,-1.0513,-1.0001) circle (2pt)  node[above]{$v_5$};
  \fill[fill=black] (0.5258,-1.0513,-1.6181) circle (2pt)  node[above left]{$v_6$};
  \draw (0,2,0) -- node[above right]{1.7468} (1.7014,-1.0513,0);
  \draw (1.7014,-1.0513,0) -- node[above]{1} (0.5258,-1.0513,1.6181);
  \draw (0.5258,-1.0513,1.6181) -- node[above]{1} (-1.3765,-1.0513,1.0001);
  \draw (-1.3765,-1.0513,1.0001) -- node[above]{1} (-1.3765,-1.0513,-1.0001);
  \draw (0.5258,-1.0513,-1.6181) -- node[above]{1} (-1.3765,-1.0513,-1.0001);
  \draw[dashed] (0.5258,-1.0513,-1.6181) --  (1.7014,-1.0513,0);
  \draw[dashed] (0,2,0) --  (0.5258,-1.0513,1.6181);
  \draw[dashed] (0,2,0) --  (-1.3765,-1.0513,1.0001);
  \draw[dashed] (0,2,0) --  (-1.3765,-1.0513,-1.0001);
  \draw[dashed] (0,2,0) --  (0.5258,-1.0513,-1.6181);
\end{tikzpicture}
\caption{An arrangement of 6 points on a unit sphere such that the minimum bottleneck path length is 1.7468.}
 \label{fig:points6} 
\end{figure} 

Finally, for 7 children, our representative is the pentagonal bipyramid whose apices are on diametrically opposite poles of the sphere, i.e., north and south pole. The other 5 vertices form a regular pentagon with a side length of approximately 1.018 and each of the five vertices is unit distance away from the south pole apex (see Figure~\ref{fig:points7}). The bottleneck length in this case is the length of an edge from the north pole apex to any one of the 5 vertices of the pentagon, which is calculated to be $\sqrt{3} \approx 1.73205$.

\begin{figure}[ht!]
\centering
\begin{tikzpicture}
  
  \fill[fill=black] (0,2,0) circle (2pt) node[below right]{$v_1$};
  \fill[fill=black] (1.7471,-0.9736,0) circle (2pt) node[above]{$v_2$};
  \fill[fill=black] (0.5399,-0.9736,1.6615) circle (2pt) node[above]{$v_3$};
  \fill[fill=black] (-1.4134,-0.9736,1.0269) circle (2pt)  node[above]{$v_4$};
  \fill[fill=black] (-1.4134,-0.9736,-1.0269) circle (2pt)  node[above]{$v_5$};
  \fill[fill=black] (0.5399,-0.9736,-1.6615) circle (2pt)  node[above left]{$v_6$};
  \fill[fill=black] (0,-2,0) circle (2pt) node[below right]{$v_7$};
  \draw (0,2,0) -- node[above right]{1.73205} (1.7471,-0.9736,0);
  \draw (1.7471,-0.9736,0) -- node[above right]{$s$} (0.5399,-0.9736,1.6615);
  \draw (0.5399,-0.9736,1.6615) -- node[above]{$s$} (-1.4134,-0.9736,1.0269);
  \draw (-1.4134,-0.9736,1.0269) -- node[above]{$s$} (-1.4134,-0.9736,-1.0269);
  \draw (-1.4134,-0.9736,-1.0269) -- node[above]{$s$} (0.5399,-0.9736,-1.6615);
  \draw (0,-2,0) -- node[above]{1} (0.5399,-0.9736,-1.6615);
  \draw[dashed] (0.5399,-0.9736,-1.6615) --  (1.7471,-0.9736,0);
  \draw[dashed] (0,2,0) --  (0.5399,-0.9736,1.6615);
  \draw[dashed] (0,2,0) --  (-1.4134,-0.9736,1.0269);
  \draw[dashed] (0,2,0) --  (-1.4134,-0.9736,-1.0269);
  \draw[dashed] (0,2,0) --  (0.5399,-0.9736,-1.6615);
  \draw[dashed] (0,-2,0) --  (1.7471,-0.9736,0);
  \draw[dashed] (0,-2,0) --  (0.5399,-0.9736,1.6615);
  \draw[dashed] (0,-2,0) --  (-1.4134,-0.9736,1.0269);
  \draw[dashed] (0,-2,0) --  (-1.4134,-0.9736,-1.0269);
\end{tikzpicture}
\caption{An arrangement of 7 points on a unit sphere such that the minimum bottleneck path length is $\sqrt{3}$. The side length $s$ is approximately 1.018.}
 \label{fig:points7} 
\end{figure}
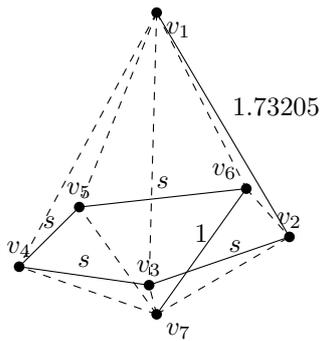

\subsection{Worst Case Ratio for the $\delta$-NKRY Algorithm for $\delta = 9,10,11$}
For the cases with 8 or more children, we were unable to obtain configurations through our computer search that approach known geometric patterns, as was the case for the previous examples. As such, we will simply present the best configurations, i.e., the configurations of children $C$ on the surface of a unit sphere with centre $v$ that yielded that largest value for $b_v(C)$, that were obtained after several searches. We do not claim that these configurations are the true worst-case configurations, in fact the likelihood of the computer search producing a solution that is sub-optimal seems to increase as the number of points increase, due to additional local optima. However, our outputs are useful for obtaining lower bounds as well as potentially giving insight into the geometric structure of the true worst-case configuration.\\~\\
For 8 children we were able to obtain a bottleneck value of approximately $1.5105$, using the point set given in Table~\ref{tab:deg9}. This point arrangement could best be described as a single point on one pole, five points almost arranged in a pentagonal shape in the opposite hemisphere, with the remaining two points situated close to the opposite pole. A plot of the point set is given in Figure~\ref{fig:pointsdeg9}.

\begin{table}[H]
\centering
\caption{A set of points on the surface of a unit sphere that has an optimal bottleneck TSP-path value of approximately $1.5105$.}
\label{tab:deg9}
\begin{tabular}{@{}ccc@{}}
\toprule
 $x$                   &  $y$                   & $z$                    \\ \midrule
0.43706   & 0.25681  & 0.86199  \\
0.19162   & 0.96778  & 0.16333  \\
-0.76029  & 0.64488  & 0.078046 \\
-0.53892  & 0.037955 & 0.8415  \\
0.57761   & -0.67051 & 0.4656  \\
-0.43604  & -0.82836 & 0.35169  \\
0.97339   & 0.22864  & 0.015302 \\
-0.082855 & -0.20074 & -0.97613 \\ \bottomrule
\end{tabular}
\end{table}

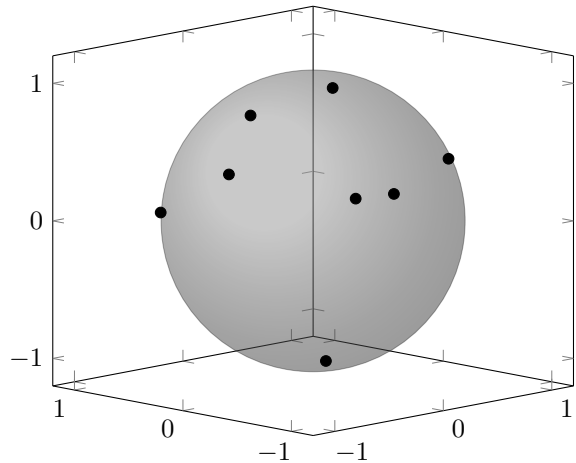
\begin{figure}[H]
\centering
\begin{tikzpicture}
\def\R{1} 
 \begin{axis}[view = {-45}{12},xmin = -1.2, xmax = 1.2,ymin = -1.2, ymax = 1.2,zmin = -1.2, zmax = 1.2]
\addplot3[black,only marks]
coordinates {
	(0.43706,0.25681,0.86199) 
	(0.19162,0.96778,0.16333)
	(-0.76029  , 0.64488  , 0.078046)
	(-0.53892  , 0.037955 , 0.8415  )
	(0.57761   , -0.67051 , 0.4656  )
	(-0.43604  , -0.82836 , 0.35169 ) 
	(0.97339   , 0.22864  , 0.015302) 
	(-0.082855 , -0.20074 , -0.97613)  
};
\filldraw[ball color=white, opacity=0.3] (axis
            cs:0,0,0)  circle [radius = 2cm];
\end{axis}
\end{tikzpicture}
\caption{The point set given in Table~\ref{tab:deg9}.}
 \label{fig:pointsdeg9} 
\end{figure} 

For 9 children, we obtained a point set with a bottleneck value of approximately $1.4095$, using the point set given in Table~\ref{tab:deg10}. This configuration could best be described as a point on one pole, 5 points in an almost pentagonal shape near the equator and 3 points forming a triangle around the opposite pole. A plot of the point set is given in Figure~\ref{fig:pointsdeg10}.\\

\begin{table}[H]
\centering
\caption{A set of points on the surface of a unit sphere that has an optimal bottleneck TSP-path value of approximately $1.4095$.}
\label{tab:deg10}
\begin{tabular}{@{}ccc@{}}
\toprule
 $x$                   &  $y$                   & $z$                    \\ \midrule
0.0016014 & -0.14111  & -0.98999 \\
-0.13233  & 0.79515   & 0.59181  \\
0.77068   & 0.33782   & 0.5403  \\
-0.87494  & -0.24761  & -0.41615 \\
-0.61366  & 0.717670   & -0.3292 \\
0.11711   & -0.94182  & 0.31506  \\
-0.83887  & 0.069402  & 0.53988  \\
0.43883   & 0.84718   & -0.29953 \\
0.87953   & -0.055801 & -0.47256 \\ \bottomrule
\end{tabular}
\end{table}

\begin{figure}[H]
\centering
\begin{tikzpicture}
\def\R{1} 

 \begin{axis}[view = {-87}{10},xmin = -1.2, xmax = 1.2,ymin = -1.2, ymax = 1.2,zmin = -1.2, zmax = 1.2]
\addplot3[black,only marks]
coordinates {
	(0.0016014 , -0.14111  , -0.98999)
	(-0.13233  , 0.79515   , 0.59181)
	(0.77068   , 0.33782   , 0.5403)
	(-0.87494  , -0.24761  , -0.41615)
	(-0.61366  , 0.717670   , -0.3292)
	(0.11711   , -0.94182  , 0.31506 )
	(-0.83887  , 0.069402  , 0.53988 )
	(0.43883   , 0.84718   , -0.29953)
	(0.87953   , -0.055801 , -0.47256)  
};
\filldraw[ball color=white, opacity=0.3] (axis
            cs:0,0,0)  circle [radius = 2.6cm];
\end{axis}
\end{tikzpicture}
\caption{The point set given in Table~\ref{tab:deg10}.}
 \label{fig:pointsdeg10} 
\end{figure}
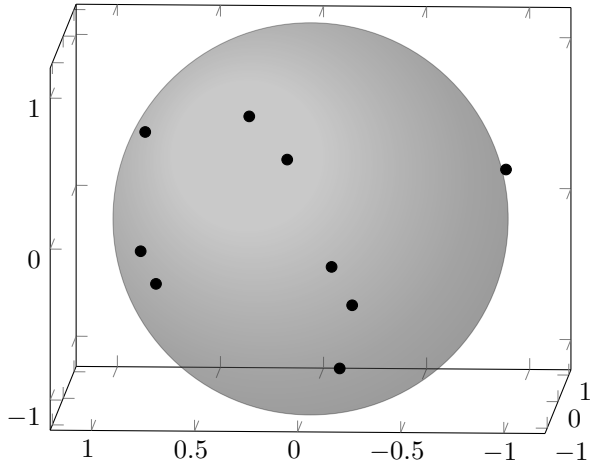 

Finally, for 10 children, we obtained a point set with a bottleneck value of approximately $1.3314$, using the point set given in Table~\ref{tab:deg11}. This configuration could best be described as a point on one pole, 6 points in an almost hexagonal shape near the equator and 3 points forming a triangle around the opposite pole. A plot of the point set is given in Figure~\ref{fig:pointsdeg11}.

\begin{table}[H]
\centering
\caption{A set of points on the surface of a unit sphere that has an optimal bottleneck TSP-path value of approximately $1.3314$.}
\label{tab:deg11}
\begin{tabular}{@{}ccc@{}}
\toprule
 $x$                   &  $y$                   & $z$                    \\ \midrule
0.9707 & 0.015168   & 0.23983   \\
0.5  & 0.86603    & 0                   \\
-0.49562 & 0.85955    & 0.12461   \\
-0.98247 & 0 & -0.1864 \\
0.5  & -0.86603   & 0                   \\
-0.49986 & -0.86579   & -0.023435 \\
0.68647  & 0.13024    & -0.7154  \\
-0.24313 & 0.55298    & -0.79693  \\
-0.10477 & -0.4889   & -0.86603  \\
-0.12718 & -0.084591  & 0.98827   \\ \bottomrule
\end{tabular}
\end{table}

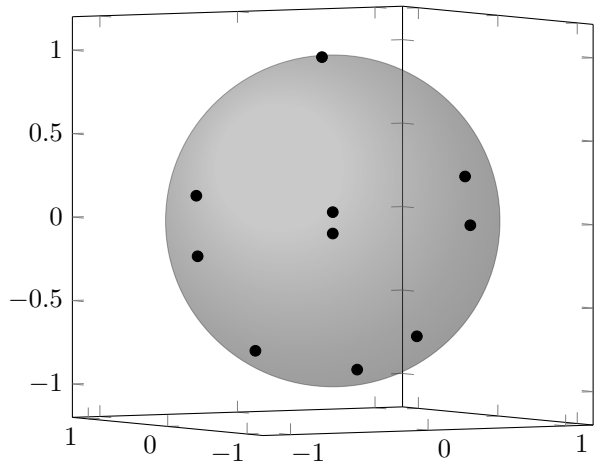
\begin{figure}[H]
\centering
\begin{tikzpicture}
\def\R{1} 

 \begin{axis}[view = {-30}{3},xmin = -1.2, xmax = 1.2,ymin = -1.2, ymax = 1.2,zmin = -1.2, zmax = 1.2]
\addplot3[black,only marks]
coordinates {
	(0.9707 , 0.015168   , 0.23983   )
	(0.5  , 0.86603    , 0          )
	(-0.49562 , 0.85955    , 0.12461   )
	(-0.98247 , 0 , -0.1864 )
	(0.5  , -0.86603   , 0  )                 
	(-0.49986 , -0.86579   , -0.023435 )
	(0.68647  , 0.13024    , -0.7154  )
	(-0.24313 , 0.55298    , -0.79693  )
	(-0.10477 , -0.4889   , -0.86603  )
	(-0.12718 , -0.084591  , 0.98827  ) 
};
\filldraw[ball color=white, opacity=0.3] (axis
            cs:0,0,0)  circle [radius = 2.2cm];
\end{axis}
\end{tikzpicture}
\caption{The point set given in Table~\ref{tab:deg11}.}
 \label{fig:pointsdeg11} 
\end{figure}

The point sets we have obtained, both analytically and experimentally, seem to suggest that as the number of children increases, the length of the bottleneck edge in a worst-case arrangement of the children decreases. We believe that this is case, and hence we make the following conjecture.

\begin{conjecture} \label{conj:monotone}
Let $v$ be a root vertex with children $\{v_1, \dots, v_k\}$, where $k \geq 2$ such that angles between children with respect to $v$ are at least $60^{\circ}$ and the radial distances of all children are equal to 1. Let $P$ be the optimal bottleneck path that starts at $v$ and visits all children and suppose that $\{v_1, \dots, v_k\}$ are placed so as to maximise the bottleneck length $b$ of $P$.\\
Also let $w$ be a root vertex with children $\{w_1, \dots, w_{k+1}\}$, such that angles between children with respect to $w$ are at least $60^{\circ}$ and the radial distances of all children are equal to 1. Let $P'$ be the optimal bottleneck path that starts at $w$ and visits all children and suppose that $\{w_1, \dots, w_{k+1}\}$ are placed so as to maximise the bottleneck length $b'$ of $P'$.\\
Then $b \geq b'$.
\end{conjecture}

If Conjecture~\ref{conj:monotone} were true, then it would imply that the star for which $\delta$-NKRY algorithm gives the worst performance will have $\delta - 1$ children, since the algorithm does not perform edge swaps when the root node has $\delta -2$ or fewer children. Hence, the stars that were proven analytically to be the worst case arrangements, namely the stars with 2, 3 and 4 children, would yield upper bounds for the performance ratios of the $\delta$-NKRY algorithm. This would in turn yield exact values for the performance ratios of the $4$-NKRY and $5$-NKRY algorithms (since the performance ratio of the $3$-NKRY algorithm has already been shown), and we could use the worst-case arrangement of $k-1$ children on the unit sphere, for any $k > 2$, to give an upper bound for the performance ratio of the $\delta$-NKRY algorithm for all $\delta \geq k$.\\~\\
The statement of Conjecture~\ref{conj:monotone} seems somewhat intuitive; the more children the root node has, the less space we have in which to spread children out from one another. Furthermore, all of our experimentation thus far seems to agree with this observation. However, we have been unable to establish a rigorous proof of Conjecture~\ref{conj:monotone} due to the difficulty of the bottleneck TSP-path problem and the lack of assumptions that can be made in regards to how the optimal path changes when local changes are applied.

\subsection{A More General Adaptation of the KRY Algorithm for the $\delta$-E3MBST Problem} \label{EPKRY}
In the previous section, we considered the naive KRY algorithm as a generalisation of the KRY algorithm. Whilst the $\delta$-NKRY is a relatively simple generalisation, this simplicity may come at the cost of accuracy, as the $\delta$-NKRY algorithm has the tendency to reduce the degrees further than what is necessary. For instance, consider the $8$-NKRY algorithm. This algorithm will perform edge swaps whenever the root node has a degree of 7 or more. However, after the edge swap, the node will have a degree of at most 3 in the final tree, due to the edge swap being a path from the root node through its children.\\~\\
To address this issue, we propose the a more sophisticated generalisation of the KRY algorithm, namely the \textit{$k$-partition} KRY algorithm for the $\delta$-E3MBST problem, or $(\delta,k)$-PKRY for short. The $(\delta,k)$-PKRY algorithm is similar to the $\delta$-NKRY in that it recursively performs edge swaps with respect to nodes of a rooted MST if the degree of the current root node is strictly greater than $\delta - 2$. The difference between this partition version and the naive version of the KRY algorithm is in the search space of possible edge swaps. For a given root node, the $\delta$-NKRY algorithm only considers edge swaps which result in a path starting at the root node and passing through all children. Of the potential paths, it chooses the optimal bottleneck TSP path. The partition version however, may consider paths but can also consider other sets of edges which depend on certain partitionings of the children of the current root node. The parameter $k$ indicates the number of sets we will partition the children into, and therefore we will find that $1 \leq k \leq \delta - 2$ in order for the degree constraint to be satisfied.\\~\\
Given a current root node with $j$ children $(\delta,k)$-PKRY considers all possible partitions of children into $k$ sets. For a given partition, the chosen edge swap consists of the union of minimum bottleneck TSP-paths that start at the root node and visit all nodes in a single partition. For example, if our root node is $v$, and the partition of children is $\{\{v_1,v_2\},\{v_3\},\{v_4,v_5\}\}$ ($k = 3$), then the edges after the edge swap could consist of $\{(v,v_1),(v_1,v_2),(v,v_3),(v,v_4),(v_4,v_5)\}$ (we would replace $(v,v_1)$ by $(v,v_2)$ or $(v,v_4)$ by $(v,v_5)$ if it reduced the length of the longest edge). See Figure~\ref{fig:PKRYexample} for an illustration of this example. Note that no node in the resultant subtree other than the root will have a degree of more than 2. The algorithm searches through all the possible partitions that use $k$ sets and chooses the one that results in edge swaps where the longest edge swapped in is minimum. As such, if an edge swap was performed with respect to a root node, the largest degree that the root node can have in the final tree is $k+2$. We present the $(\delta,k)$-PKRY algorithm as Algorithm~\ref{alg:PKRY}.\\

\begin{algorithm}[htb]
\caption{: $(\delta,k)$-PKRY} \label{alg:PKRY}
\begin{tabbing}
  ....\=....\=....\=....\=................... \kill \\ [-2ex]
\textbf{Input:} A rooted tree $T$ over a point set $P$ with root $v$ and a partially built\\ solution $T^*$.\\
\\
\textbf{if} $v$ has at least one child\\
\> Let the children of $v$ be $v_1, \dots, v_c$, where $c$ is the number of children of $v$.\\
\> \textbf{if} $c \geq \delta-1$,\\
\>\> \textbf{For each} partition $Q$ of $\{v_1, \dots ,v_c\}$ into $k$ subsets,\\
\>\>\> Let $C_i$ be the $i$-th subset of $Q$.\\
\>\>\> Let $P_{Q,i}$ be the optimal bottneck TSP-path through $C_i$ starting at $v$.\\
\>\>\> Let $b_Q = \max_{i = 1}^{k}(b_v(C_i))$, where $b_v(C_i)$ denotes the length of the longest\\ \>\>\> edge in $P_{Q,i}$.\\
\>\> Choose the partition $Q^*$ such that $b_{Q^*}$ is minimum.\\
\>\> \textit{For} $i = 1, \dots ,k$,\\
\>\>\> Add the edges of $P_{Q^*,i}$ to $T^*$.\\
\> \textbf{else} \\
\>\> Add each edge $(v,v_i)$ to $T^*$, where $i = 1, \dots, c$.\\
\> Perform Algorithm~\ref{alg:PKRY} with $T:= T_{v_i}$ as input for $i = 1, \dots, c$, where $T_{v_i}$ denotes\\ \> the subtree rooted at $v_i$.\\
\\
\textbf{Output:} $T^*$.
\end{tabbing}
\end{algorithm}

\begin{figure}[ht!]  
  \begin{subfigure}[b]{0.5\linewidth}
    \centering
    \includegraphics[scale =0.7]{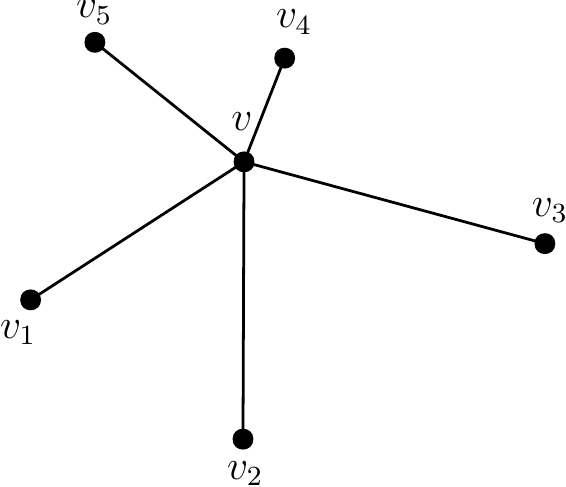} 
    
  \end{subfigure}
  \hspace{4ex}
  \begin{subfigure}[b]{0.5\linewidth}
    \centering
    \includegraphics[trim = 0 -0 0 0,clip,scale =0.7]{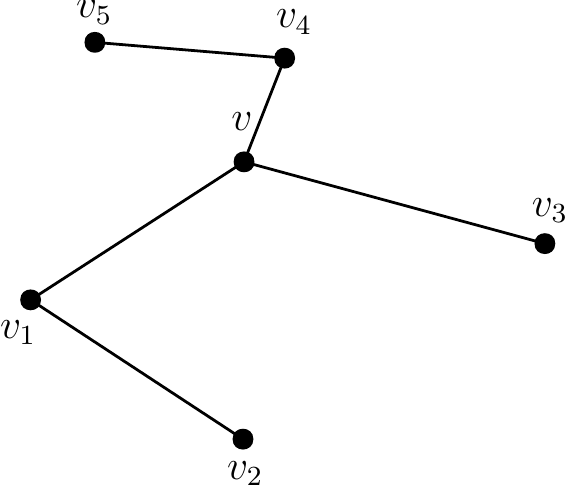} 
  \end{subfigure}
  \caption{An example edge swap for the $(\delta,k)$-PKRY algorithm, where $k=3$.} 
    \label{fig:PKRYexample} 
\end{figure}

Note that the $\delta$-NKRY algorithm is the same as the $(\delta,1)$-PKRY algorithm since it only ever considers the partition of children into a single set. One can also observe that the algorithm improves in performance as $k$ increases, a fact which we make explicit in the following lemma.

\begin{lemma}\label{lem:partitionlemma}
Given a point set $P$, the length of the longest edge in the output tree of $(\delta,k-1)$-PKRY applied to $P$ is an upper bound for the length of the longest edge in the output tree of $(\delta,k)$-PKRY applied to $P$, where $1<k\leq \delta - 2$ and $\delta \geq 4$.
\end{lemma}
\begin{proof}
Consider a root node $v$ in the rooted MST for $P$ such that $(\delta,k-1)$-PKRY will employ local edge swaps with respect to $v$. This means that the number of children of $v$ is at least $\delta - 1$, a number which is strictly greater than $k - 1$. Hence, whichever edges $(\delta,k-1)$-PKRY chooses to swap in or out, we can guarantee that there will exist a path starting at $v$ through two or more children of $v$ after the swaps. Choose such a path and remove the endpoint that was not $v$, say $v_i$, by deleting the edge between $v_i$ and its neighbour in the path, which we will denote $v_{i-1}$. Since $v_i$ is now isolated, add the edge $(v,v_i)$ to complete the subtree. This new subtree is a feasible output of the $(\delta,k)$-PKRY algorithm since it uses $k$ partitions (the newly created set in the partition is $\{v_i\}$). Additionally, the new subtree cannot have a greater bottleneck value than the previous subtree since this would imply that $d(v,v_i) > d(v_{i-1},v_i)$, which contradicts the optimality of the original MST. Thus the result is proven. 
\end{proof}

An obvious drawback of the $(\delta,k)$-PKRY algorithm is the size of the search space. Whilst the sizes of the search space are always technically bounded above by a constant since the degrees of vertices in an E3MST are bounded above by a constant, they can be considerably large. This is because the sizes of our search spaces are Stirling numbers of the second kind which grow exponentially in the number of children. For instance, the number of ways of partitioning $10$ children into $5$ sets is the Stirling number $S(10,5)$ which is equal to $42525$. For the $(\delta,k)$-PKRY algorithm to be implemented in a practical setting, an appropriate value of $k$ should be chosen that balances both the need for accuracy and the need for time efficiency.\\\\

\subsection{Performance Ratio of the $(\delta,k)$-PKRY Algorithm}
As in the previous subsections, we wish to perform some analysis of the worst case point arrangements for the $(\delta,k)$-PKRY algorithm. Once again we need only consider the worst-case local arrangement of children around a root vertex due to recursive nature of this bottleneck algorithm. In addition, Theorem~\ref{thm:equidistant} still applies and we can once again assume that the worst case arrangement occurs when the children are arranged on the surface of the unit sphere. However unlike previously, we will not consider all possible values of $\delta$. This is because the large search spaces present when $k >1$ result in our computational experimentation being inefficient in both computation time and in accuracy. Similarly, we will not be considering all possible values of $k$. When $k = 2$ and we perform a swap with respect to a partition of $3$ children (hence $\delta =4$, since $\delta =3$ violates the upper bound for $k$ and we would not perform the swap for $\delta \geq 5$), we have the following result.

\begin{theorem}
The largest possible bottleneck increase resulting from edge swaps with respect to a root node with 3 children for the $(4,2)$-PKRY algorithm is a factor of $\sqrt{3}$.
\end{theorem} 
\begin{proof}
Let $v$ be the root node with children $v_1,v_2,v_3$. As before, we assume that $v$ lies at the centre of a unit sphere with its children on the surface such that every child is at least unit distance away from every other child. Since we can use up to two partitions, we will only need to swap in exactly one of the edges of $\{(v_1,v_2),(v_2,v_3),(v_1,v_3)\}$. This is because $v$ is allowed to have a degree of 2 in the subtree after the swaps. Hence, the algorithm will always choose the shortest edge from $\{(v_1,v_2),(v_2,v_3),(v_1,v_3)\}$ to swap in and so the worst case will occur when the length of the shortest edge in the triangle $(v_1,v_2,v_3)$ is maximised. It is not difficult to see that the smallest edge will be maximised when the triangle is equilateral. The largest possible equilateral triangle will is one in which the vertices are placed around an equator of the sphere, and the edge length of such a triangle is $\sqrt{3}$.
\end{proof}
We know from previous results that if $k = 1$, then the worst case for three children would represent an increase of the bottleneck length by a factor of approximately 1.931. Therefore an increase in $k$ has improved the performance of the algorithm for three children.\\~\\ 
We performed computational experiments for the cases with $\delta =5$ (4 children), i.e., $(5,2)$-PKRY and $(5,3)$-PKRY.
When there are 4 children and $k = 2$, the worst-case obtained by our computational search is the point set presented in Table~\ref{tab:4part2} and Figure~\ref{fig:4part2} and yields bottleneck value of approximately 1.6829. When there are 4 children and $k = 3$, the worst-case obtained by our computational search is the point set presented in Table~\ref{tab:4part3} and Figure~\ref{fig:4part3} and yields bottleneck value of approximately 1.6330. 

\begin{table}[H]
\centering
\caption{A set of 4 points on the surface of a unit sphere such that the approximation factor for the $(5,2)$-PKRY algorithm is 1.6829.}
\label{tab:4part2}
\begin{tabular}{@{}ccc@{}}
\toprule
 $x$                   &  $y$                   & $z$                    \\ \midrule
0.67484 & -0.60944  & -0.41615 \\
-0.90907 & 0.020503   & -0.41615  \\
0.30056   & 0.85073   & 0.43118  \\
0  & 0  & 1 \\ \bottomrule
\end{tabular}
\end{table}

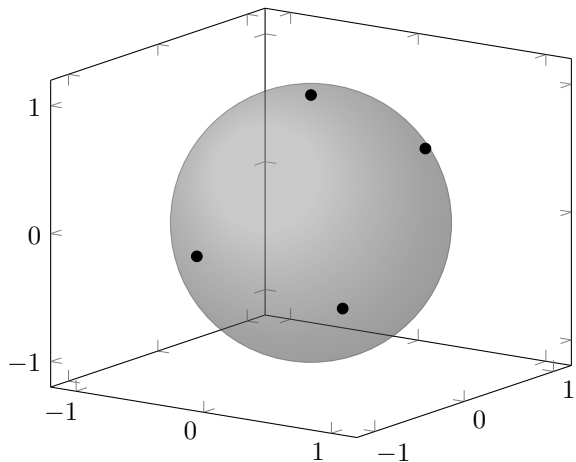
\begin{figure}[H]
\centering
\begin{tikzpicture}
\def\R{1} 

 \begin{axis}[view = {35}{16},xmin = -1.2, xmax = 1.2,ymin = -1.2, ymax = 1.2,zmin = -1.2, zmax = 1.2]
\addplot3[black,only marks]
coordinates {
	(0.67484 , -0.60944  , -0.41615)
	(-0.90907 , 0.020503   , -0.41615)
	(0.30056   , 0.85073   , 0.43118)
	(0  , 0  , 1)
	
};
\filldraw[ball color=white, opacity=0.3] (axis
            cs:0,0,0)  circle [radius = 1.85cm];
\end{axis}
\end{tikzpicture}
\caption{The point set given in Table~\ref{tab:4part2}.}
 \label{fig:4part2} 
\end{figure}

\begin{table}[H]
\centering
\caption{A set of 4 points on the surface of a unit sphere such that the approximation factor for the $(5,3)$-PKRY algorithm is 1.6330.}
\label{tab:4part3}
\begin{tabular}{@{}ccc@{}}
\toprule
 $x$                   &  $y$                   & $z$                    \\ \midrule

-0.70503 & 0.32904 & -0.62822\\
0.26587 & -0.92974 & -0.25475\\
0.798 & 0.60054 & -0.050429\\
-0.35884 & 0.00015484 & 0.9334 \\ \bottomrule
\end{tabular}
\end{table}

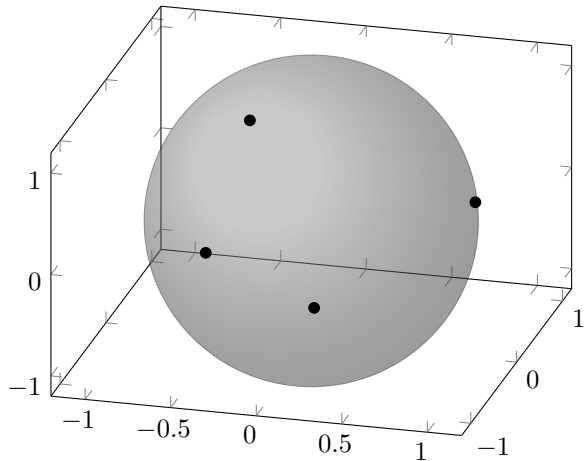
\begin{figure}[H]
\centering
\begin{tikzpicture}
\def\R{1} 

 \begin{axis}[view = {15}{32},xmin = -1.2, xmax = 1.2,ymin = -1.2, ymax = 1.2,zmin = -1.2, zmax = 1.2]
\addplot3[black,only marks]
coordinates {
	(-0.70503 , 0.32904 , -0.62822)
	(0.26587 , -0.92974 , -0.25475)
	(0.798 , 0.60054 , -0.050429)
	(-0.35884 , 0.00015484 , 0.9334)
	
};
\filldraw[ball color=white, opacity=0.3] (axis
            cs:0,0,0)  circle [radius = 2.2cm];
\end{axis}
\end{tikzpicture}
\caption{The point set given in Table~\ref{tab:4part3}.}
 \label{fig:4part3} 
\end{figure}

Finally, we performed computational experiments for the $\delta = 6$ cases. For 5 children, the worst case bottleneck values obtained were 1.6330, 1.4991 and 1.3834 for $k = 2,3$ and $4$ respectively. These are presented in Table \ref{tab:5part2}, \ref{tab:5part3} and \ref{tab:5part4}, and Figure \ref{fig:5part2}, \ref{fig:5part3} and \ref{fig:5part4}.

\begin{table}[H]
\centering
\caption{A set of 5 points on the surface of a unit sphere such that the approximation factor for the $(6,2)$-PKRY algorithm is 1.6330.}
\label{tab:5part2}
\begin{tabular}{@{}ccc@{}}
\toprule
 $x$                   &  $y$                   & $z$                    \\ \midrule

0.14108 &	-0.0031796 &	-0.98999\\
0.5301 &	-0.77315 &	-0.34818\\
0.47115 &	0.85383 &	-0.22133\\
-0.90627 &	0.00011443 &	-0.4227\\
-0.094952 &	-0.080796 &	0.9922\\ \bottomrule
\end{tabular}
\end{table}

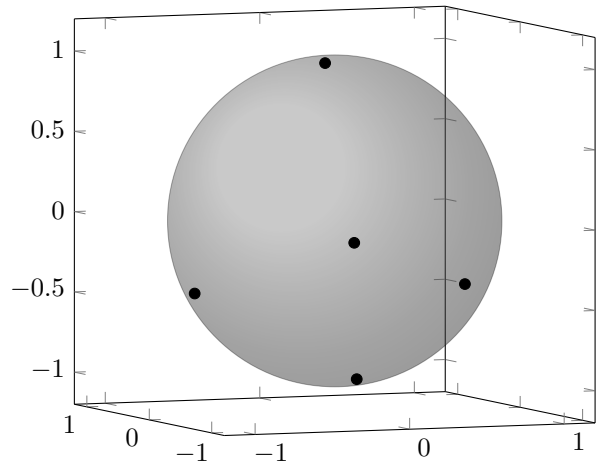
\begin{figure}[H]
\centering
\begin{tikzpicture}
\def\R{1} 

 \begin{axis}[view = {-22}{5},xmin = -1.2, xmax = 1.2,ymin = -1.2, ymax = 1.2,zmin = -1.2, zmax = 1.2]
\addplot3[black,only marks]
coordinates {
	(0.14108 ,	-0.0031796 ,	-0.98999)
	(0.5301 ,	-0.77315 ,	-0.34818)
	(0.47115 ,	0.85383 ,	-0.22133)
	(-0.90627 ,	0.00011443 ,	-0.4227)
	(-0.094952 ,	-0.080796 ,	0.9922)
	
};
\filldraw[ball color=white, opacity=0.3] (axis
            cs:0,0,0)  circle [radius = 2.2cm];
\end{axis}
\end{tikzpicture}
\caption{The point set given in Table~\ref{tab:5part2}.}
 \label{fig:5part2} 
\end{figure}

\begin{table}[H]
\centering
\caption{A set of 5 points on the surface of a unit sphere such that the approximation factor for the $(6,3)$-PKRY algorithm is 1.4991.}
\label{tab:5part3}
\begin{tabular}{@{}ccc@{}}
\toprule
 $x$                   &  $y$                   & $z$                    \\ \midrule

0.25438 &	-0.72417 &	0.641\\
0.07693 &	0.97765 &	-0.19566\\
-0.88114 &	-0.35965 &	-0.30698\\
0.54273 &	-0.32431 &	-0.77477\\
-0.23211 &	0.08564 &	0.96891\\ \bottomrule
\end{tabular}
\end{table}

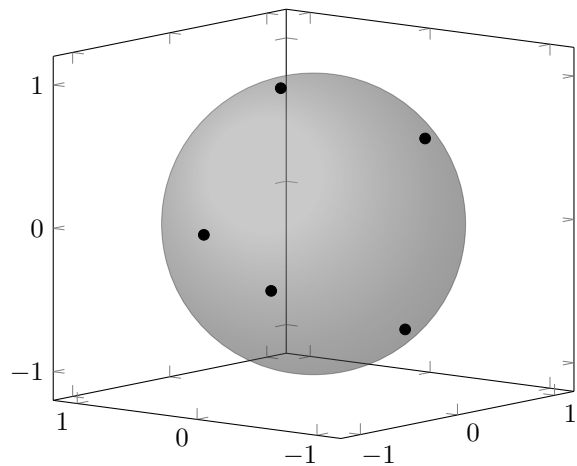
\begin{figure}[H]
\centering
\begin{tikzpicture}
\def\R{1} 

 \begin{axis}[view = {-51}{10},xmin = -1.2, xmax = 1.2,ymin = -1.2, ymax = 1.2,zmin = -1.2, zmax = 1.2]
\addplot3[black,only marks]
coordinates {
	(0.25438 ,	-0.72417 ,	0.641)
	(0.07693 ,	0.97765 ,	-0.19566)
	(-0.88114 ,	-0.35965 ,	-0.30698)
	(0.54273 ,	-0.32431 ,	-0.77477)
	(-0.23211 ,	0.08564 ,	0.96891)
	
};
\filldraw[ball color=white, opacity=0.3] (axis
            cs:0,0,0)  circle [radius = 2cm];
\end{axis}
\end{tikzpicture}
\caption{The point set given in Table~\ref{tab:5part3}.}
 \label{fig:5part3} 
\end{figure}

\begin{table}[H]
\centering
\caption{A set of 5 points on the surface of a unit sphere such that the approximation factor for the $(6,4)$-PKRY algorithm is 1.3834.}
\label{tab:5part4}
\begin{tabular}{@{}ccc@{}}
\toprule
 $x$                   &  $y$                   & $z$                    \\ \midrule

0.16008 &	0.98616 &	0.043096\\
0.90272 &	-0.084403 &	-0.42186\\
-0.94994 &	0.20809 &	-0.23307\\
-0.17824 &	-0.93847 &	-0.2958\\
0 &	0 &	1\\ \bottomrule
\end{tabular}
\end{table}

\begin{figure}[H]
\centering
\begin{tikzpicture}
\def\R{1} 

 \begin{axis}[view = {8}{16},xmin = -1.2, xmax = 1.2,ymin = -1.2, ymax = 1.2,zmin = -1.2, zmax = 1.2]
\addplot3[black,only marks]
coordinates {
	(0.16008 ,	0.98616 ,	0.043096)
	(0.90272 ,	-0.084403 ,	-0.42186)
	(-0.94994 ,	0.20809 ,	-0.23307)
	(-0.17824 ,	-0.93847 ,	-0.2958)
	(0 ,	0 ,	1)
	
};
\filldraw[ball color=white, opacity=0.3] (axis
            cs:0,0,0)  circle [radius = 2.45cm];
\end{axis}
\end{tikzpicture}
\caption{The point set given in Table~\ref{tab:5part4}.}
 \label{fig:5part4} 
\end{figure}
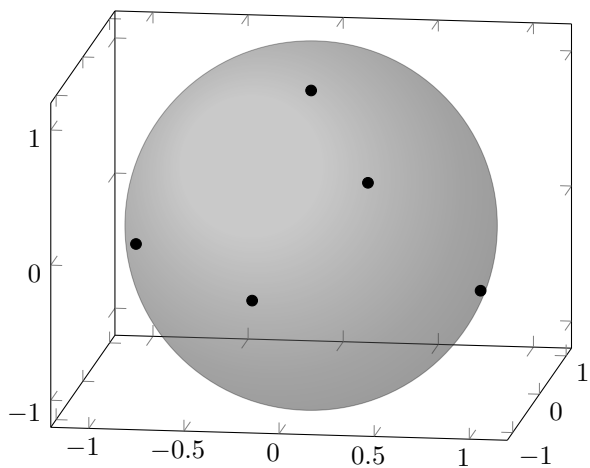

As in the previous section, we also conjecture that there is a monotonicity with respect the length of the bottleneck edge in the worst case arrangements of $m$ points, i.e., the bottleneck value of the worst-case arrangement of $m$ children for the $(\delta,k)$-PKRY algorithm is an upperbound for the bottleneck value of the worst-case arrangement of $m+1$ children for the same algorithm. We present this statement as the following conjecture, which is a stronger version of Conjecture~\ref{conj:monotone}.

\begin{conjecture} \label{conj:monotone2}
Let $S_m$ be a star with $m$ children rooted at the centre vertex, such that $S_m$ maximises the length of the longest edge in the tree given by the $(\delta,k)$-PKRY algorithm when given $S_m$ as input, where the children of $S_m$ are unit distance from the centre vertex and no pair of children have a distance of less than 1 from one another. Let $S_{m+1}$ be defined similarly, but with $m+1$ children.\\ 
Let $b$,$b'$ be the bottleneck values of the outputs of $(\delta,k)$-PKRY when performed on $S_m$ and $S_{m+1}$ respectively. Then $b \geq b'$.
\end{conjecture}

We did not obtain results for $\delta \geq 7$ due the Stirling numbers becoming too large, however our results thus far seem to provide evidence for the monotonicity statement of Conjecture~\ref{conj:monotone2} whilst also providing evidence for the intuition that increasing the number of sets in the partition improves the performance of the $(\delta,k)$-PKRY algorithm.

\section{Approximation Algorithms for the $\delta$-R3MBST Problem}
In this section, we describe and analyse ways of extending the KRY algorithm for application in 3-dimensional space under the 3-dimensional rectilinear metric, $L_1$. For two points $p_1 = (x_1,y_1,z_1)$ and $p_2 = (x_2,y_2,z_2)$ in $R^3$, the rectilinear or Manhattan distance between the points is \[d_R(p_1,p_2) = |x_1-x_2|+|y_1-y_2|+|z_1-z_2|.\]
Since all other aspects of the problem besides the distance metric remain unchanged, our algorithms will be identical in implementation for the rectilinear and Euclidean version of the problem.\\~\\
First we consider the worst case performances of the $\delta$-NKRY algorithm in rectilinear space. In the previous section, it was determined that the worst case instances for the $\delta$-NKRY algorithm in Euclidean space could be assumed to be stars with $\delta -1$ children, each child being equidistant from the centre vertex. This assumption is reliant on Lemma~\ref{lem:sidez}, and so in order to make equivalent assumptions for the rectilinear  version of the problem, we need an analogous result in rectilinear space. We present such a result as the following lemma.

\begin{lemma}
Let $v,v_1,v_2$ be points in $R^3$ such that $d_R(v_1,v_2) \geq \max(d_R(v,v_1),$ $d_R(v,v_2))$. Then moving the point $v_1$ away from $v$ in the direction of the line through $v$ and $v_1$ will not decrease the distance $d_R(v_1,v_2)$.
\end{lemma}
\begin{proof}
Let $l_2 = d_R(v,v_2)$. Consider the ball around $v_2$ with radius $l_2$, i.e., consider the set of points $\{p:d_R(p,v_2) \leq l_2 \}$. In 3-dimensional rectilinear space, such a set defines an octahedron centred at $v_2$. Since $d_R(v,v_2) = l_2$, the point $v$ will lie on the boundary of the octahedron, and since $d_R(v_1,v_2) \geq \max(d_R(v,v_1),d_R(v,v_2))$, $v_1$ must lie either outside or on the boundary the octahedron. Then, by convexity, if we were to move $v_1$, away from $v$ in the direction of the line through $v$ and $v_1$, then we would either be moving $v_1$ along the boundary of the octahedron or away from the octahedron. In either case, $v_1$ would not be moving closer to the centre of the octahedron and so the distance $d_R(v_1,v_2)$ would not decrease. 
\end{proof}

Since we know from Section~\ref{sec:Hadwiger} that when $\delta \geq 14$, there always exists an RMBST whose maximum degree is no more than $\delta$, we could mimic the approach of the previous section to find the worst case point arrangements for the $\delta$-NKRY. That is, for each $k \in \{2,\dots ,13\}$ we could attempt to find a set of $k$ points arranged in the surface of a unit ball (in this case an octahedron) such that the longest edge in a minimum bottleneck TSP-path through the points is of maximum length and the distance between every pair of points is at least 1. However, it turns out that for this rectilinear variant of the problem, such a strategy is unnecessary. This is because the length of the bottleneck edge does not decrease monotonically as $k$ increases, as appeared to be the case for the Euclidean version. Instead, it always possible to arrange the sets of points such that the longest edge is of length 2. We give the points in the following theorem.

\begin{theorem}
The worst case performance ratio for the $\delta$-NKRY algorithm in 3-dimensional rectilinear space is 2, for every $\delta \in \{3, \dots, 13\}$.
\end{theorem}
\begin{proof}
Suppose we have a unit octahedron in 3-dimensional rectilinear space, whose centre lies at the origin. The following point set on the surface of the octahedron consists of 13 points such that minimum bottleneck path through the points must use an edge of length 2, and the distance between every pair of points is at least 1;

\begin{table}[H]
\centering
\begin{tabular}{cc}
$(0.5,0.5,0)$ & $(0,0,1)$ \\
$(-0.5,0.5,0)$ & $(1,0,0)$ \\
$(-0.5,-0.5,0)$ & $(0,1,0)$ \\
$(0.5,-0.5,0)$ & $(-1,0,0)$ \\
$(0.5,0,-0.5)$ & $(0,-1,0)$ \\
$(0,0.5,-0.5)$ & $(0,0,-1)$ \\
$(-0.5,0,-0.5)$ & 
\end{tabular}
\end{table}

This point set consists of all 6 vertices of the octahedron and midpoints of 7 of the 12 edges of the octahedron. In this arrangement, the point $(0,0,1)$ is always incident to the bottleneck edge. Note that we could extend this set to 14 points by adding another midpoint, $(0,-0.5,-0.5)$, and the point set would still have the same properties as before. The 14 point set is shown in Figure~\ref{fig:points14}.
\end{proof}

\begin{figure}[ht!]
\centering
\begin{tikzpicture}
  
  \fill[fill=black] (0,2,0) circle (2pt);
  \fill[fill=black] (2,0,0) circle (2pt);
  \fill[fill=black] (0,0,2) circle (2pt);
  \fill[fill=black] (-2,0,0) circle (2pt);
  \fill[fill=black] (0,0,-2) circle (2pt);
  \fill[fill=black] (1,0,1) circle (2pt);
  \fill[fill=black] (-1,0,1) circle (2pt);
  \fill[fill=black] (-1,0,-1) circle (2pt);
  \fill[fill=black] (1,0,-1) circle (2pt);
  \fill[fill=black] (1,-1,0) circle (2pt);
  \fill[fill=black] (0,-1,1) circle (2pt);
  \fill[fill=black] (-1,-1,0) circle (2pt);
  \fill[fill=black] (0,-1,-1) circle (2pt);
  \fill[fill=black] (0,-2,0) circle (2pt);
  
  \draw (0,2,0) -- (2,0,0);
  \draw (0,2,0) -- (0,0,2);
  \draw (0,2,0) -- (-2,0,0);
  \draw (0,2,0) -- (0,0,-2);
  
  \draw (2,0,0) -- (0,0,2);
  \draw (0,0,2) -- (-2,0,0);
  \draw (-2,0,0) -- (0,0,-2);
  \draw (0,0,-2) -- (2,0,0);
  
  \draw (0,-2,0) -- (2,0,0);
  \draw (0,-2,0) -- (0,0,2);
  \draw (0,-2,0) -- (-2,0,0);
  \draw (0,-2,0) -- (0,0,-2);
  

\end{tikzpicture}
\caption{An arrangement of 14 points on a unit octahedron such that the minimum rectilinear bottleneck path length is 2.}
 \label{fig:points14} 
\end{figure}
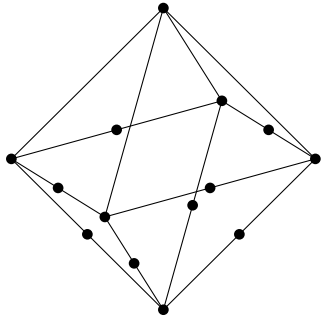 

Finally, we analyse the application of the generalised version of the KRY algorithm, the $(\delta, k)$-PKRY algorithm from Section~\ref{EPKRY}, to the rectilinear version of the problem. As in the case of $\delta$-NKRY, we also obtain a result which eliminates the need to consider certain cases, although it does not cover all the cases.

\begin{theorem}\label{thm:rect6}
The worst case performance ratio for the $(\delta,k)$-PKRY algorithm in 3-dimensional rectilinear space is 2, for every $\delta \in \{3, \dots, 7\}$, where $k \in [1,\delta -2 ]$.
\end{theorem}
\begin{proof}
Consider the point set containing 6 points in which each point is placed at a different vertex of the unit octahedron. Every point in this configuration is diametrically opposite to every other point, hence $(\delta,k)$-PKRY will yield an approximation factor of 2 for $\delta \in \{3, \dots, 7\}$ and $k \in [1,\delta -2 ]$ when given the rooted star implied by this configuration.   
\end{proof}

By adding points to the configuration of Theorem~\ref{thm:rect6}, we obtain some results for larger values of $\delta$.

\begin{proposition}\label{lem:rect7}
The worst case performance ratio for the $(8,k)$-PKRY algorithm in 3-dimensional rectilinear space is 2, for $k \in \{2,3,4 \}$.
\end{proposition}
\begin{proof}
Consider the point set in Theorem~\ref{thm:rect6} on the surface of a unit octahedron. If we add a point at coordinates $(0.5,0,0.5)$, then the largest number of vertices we can have in a path that does not use an edge of length 2 is 3 (it contains the vertices $(0,0,1),(0.5,0,0.5)$ and $(1,0,0)$). Hence if $k \leq 4$, the $(8,k)$-PKRY algorithm must use an edge of length 2. 
\end{proof}

\begin{proposition}\label{lem:rect8}
The worst case performance ratio for the $(9,k)$-PKRY algorithm in 3-dimensional rectilinear space is 2, for $k \in \{2,3 \}$.
\end{proposition}
\begin{proof}
Consider the point set with seven points on the surface of the unit octahedron from Proposition~\ref{lem:rect7}. If we add the point $(-0.5,0,0.5)$ to the set, then path that uses the largest number of vertices and no edges of length 2 is the path through $(-1,0,0),(-0.5,0,0.5),(0,0,1),(0.5,0,0.5),(1,0,0)$ which uses 5 vertices. Hence if $k \leq 3$, the $(9,k)$-PKRY algorithm must use an edge of length 2. 
\end{proof}

\begin{proposition}\label{lem:rect10}
The worst case performance ratio for the $(\delta,2)$-PKRY algorithm in 3-dimensional rectilinear space is 2, for $\delta \in \{10,11\}$.
\end{proposition}
\begin{proof}
Consider the point set with eight points on the surface of the unit octahedron from Proposition~\ref{lem:rect8}. If we add the points $(0.5,0,-0.5)$ and $(-0.5,0,-0.5)$ to the set, then path that uses the largest number of vertices and no edges of length 2 is the path through $(-1,0,0)$, $(-0.5,0,0.5)$, $(0,0,1)$, $(0.5,0,0.5)$, $(1,0,0)$, $(0.5,0,-0.5)$, $(0,0,-1)$, $(-0.5,0,-0.5)$ which uses 8 vertices. Hence if $k \leq 2$, both $(10,k)$-PKRY and $(11,k)$-PKRY must use an edge of length 2. 
\end{proof}

Since so many values of $\delta$ and $k$ can be proven to yield the same approximation ratio of 2, the next obvious question is whether or not this trend continues for the other values of $\delta$ and $k$. Whilst prohibitively large Stirling numbers prevented us from investigating all the remaining cases, we did manage to find a case for which our computational search was unable to produce a local arrangement of points that yields an approximation ratio of 2 for the $(\delta,k)$-PKRY algorithm.\\~\\
When $\delta = 8$ and $k = 6$, the arrangement of 7 points on the unit octahedron found by our search to have the greatest approximation factor for the $(8,6)$-PKRY algorithm is presented in Table~\ref{tab:8part6Rec}. The approximation ratio given by this set is approximately 1.2732.

\begin{table}[H]
\centering
\caption{A set of 7 points on the surface of a unit octahedron such that the approximation factor for the $(8,6)$-PKRY algorithm is 1.2732.}
\label{tab:8part6Rec}
\begin{tabular}{@{}ccc@{}}
\toprule
 $x$                   &  $y$                   & $z$                    \\ \midrule
0.30255 &	0.33407 &	0.36338  \\
-0.024044 &	0.70272 &	-0.27324  \\
-0.079598 &	0.010543 &	-0.90986 \\
0.8282 &	-0.12673 &	0.04507  \\
-0.57261 &	-0.064006 &	0.36338  \\
-0.19157 &	-0.53519 &	-0.27324  \\
0 &	0 &	1 \\ \bottomrule
\end{tabular}
\end{table}

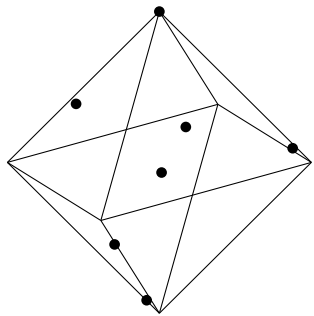
\begin{figure}[ht!]
\centering
\begin{tikzpicture}
  
  \fill[fill=black] (0.6051 ,   0.7268  ,  0.6681) circle (2pt);
  \fill[fill=black] (-0.0481  , -0.5465  ,  1.4054) circle (2pt);
  \fill[fill=black] (-0.1592 ,  -1.8197  ,  0.0211) circle (2pt);
  \fill[fill=black] (1.6564  ,  0.0901 ,  -0.2535) circle (2pt);
  \fill[fill=black] (-1.1452  ,  0.7268 ,  -0.1280) circle (2pt);
  \fill[fill=black] ( -0.3831 ,  -0.5465 ,  -1.0704) circle (2pt);
  \fill[fill=black] (0   , 2.0000      ,   0) circle (2pt);
  
  \draw (0,2,0) -- (2,0,0);
  \draw (0,2,0) -- (0,0,2);
  \draw (0,2,0) -- (-2,0,0);
  \draw (0,2,0) -- (0,0,-2);
  
  \draw (2,0,0) -- (0,0,2);
  \draw (0,0,2) -- (-2,0,0);
  \draw (-2,0,0) -- (0,0,-2);
  \draw (0,0,-2) -- (2,0,0);
  
  \draw (0,-2,0) -- (2,0,0);
  \draw (0,-2,0) -- (0,0,2);
  \draw (0,-2,0) -- (-2,0,0);
  \draw (0,-2,0) -- (0,0,-2);
  

\end{tikzpicture}
\caption{The point set given in Table~\ref{tab:8part6Rec}.}
 \label{fig:8part6Rec} 
\end{figure} 

\section{Conclusion}
In this paper, we analysed the complexity of the $\delta$-E3MBST and $\delta$-R3MBST problems. We first proved that the $5$-E3MST problem is NP-complete, then modified the proof for the bottleneck versions, to establish that the $\delta$-E3MBST and $\delta$-R3MBST problems are NP-complete for $\delta \in \{2,3,4,5\}$. We showed that the $\delta$-E3MBST problem, where $\delta \in \{4,5\}$, cannot be approximated within a factor of 1.0009 if $P \neq NP$, and we established an equivalent inapproximability constant of 1.2 for the rectilinear versions. It still remains an open problem to show that the $4$-E2MBST is NP-complete \cite{andersen2016minimum}, and if such a proof exists, we hypothesise that it could be extended to 3 dimensions to show that the $\delta$-E3MBST problem is NP-complete for $\delta = 6$ or larger.\\~\\
We also described approximation algorithms for the $\delta$-E3MBST and $\delta$-R3MBST problems; the simple $\delta$-NKRY algorithm, and the more general $(\delta,k)$-PKRY algorithm. We then analysed the algorithms by considering arrangements of points on a unit sphere that would give the worst possible performance ratios for the algorithms. By finding these pathological instances through either analytic or experimental means, we were able to give lower bounds of the performance ratios of the algorithms. For $\delta$-NKRY in 3-dimensional Euclidean space, we showed that the worst case performance ratio for $\delta =3$ is 2, and we were able to give lower bounds for each $\delta \in \{3, \dots, 11 \}$. We also showed that $\delta$-NKRY yields a worst-case performance ratio of 2 for the $\delta$-R3MBST problem for each $\delta \in \{3, \dots, 13\}$. For $(\delta,k)$-PKRY in 3-dimensional Euclidean space, we were able to give lower bounds for  $\delta \in \{4,5,6\}$ for each valid value of $k$. In 3-dimensional rectilinear space, we showed that $(\delta,k)$-PKRY gives a worst-case performance ratio of 2 for $\delta \in \{3, \dots, 7\}$ for each valid $k$. Additionally, we showed that $(\delta,k)$-PKRY gives a worst-case performance ratio of 2 in 3-dimensional rectilinear space for $\delta = 8$ and $k \in \{2,3,4\}$, for $\delta = 9$ and $k \in \{2,3\}$, and for $\delta \in \{10,11\}$ and $k =2$. If Conjecture~\ref{conj:monotone} is true, then the worst-case instances that were proven analytically would also provide upper bounds for the performance ratios of the $\delta$-NKRY algorithm, and similarly for Conjecture~\ref{conj:monotone2} and the $(\delta,2)$-PKRY algorithm. In particular, this would give exact worst-case performance ratios for the $\delta$-NKRY algorithm for $\delta \in \{4,5\}$ and the $(4,2)$-PKRY algorithm. Our results seem to corroborate these conjectures, and we hope that they are able to be proven in the future. It would also be worth having analytic proofs for the worst case-arrangements which were suggested by our computational experiments. The study of these worst-case TSP path instances may also be an interesting direction for future research.

\bibliographystyle{apalike}
\bibliography{ref}

\end{document}